\documentclass[11pt,letterpaper]{article}
\usepackage[in]{fullpage}
\usepackage{amsmath,amsthm,amsfonts,amssymb}
\usepackage{liyang}

\usepackage{framed}
\usepackage{verbatim}
\usepackage{enumitem}
\usepackage{array}
\usepackage{multirow}
\usepackage{afterpage}

\usepackage{caption} 

\usepackage[usenames,dvipsnames]{xcolor}

\makeatletter
\newtheorem*{rep@theorem}{\rep@title}
\newcommand{\newreptheorem}[2]{
\newenvironment{rep#1}[1]{
 \def\rep@title{#2 \ref{##1}}
 \begin{rep@theorem}\itshape}
 {\end{rep@theorem}}}
\makeatother
\theoremstyle{plain}

\def\colorful{0}

\ifnum\colorful=1

\newcommand{\red}[1]{{\color{red} {#1}}}
\newcommand{\green}[1]{{\color{green} {#1}}}

\fi
\ifnum\colorful=0

\newcommand{\red}[1]{{{#1}}}
\newcommand{\green}[1]{{{#1}}}

\fi

\usepackage{boxedminipage}

\newcommand{\ignore}[1]{}

\newcommand{\stars}{\mathrm{stars}} 

\title{Adaptivity is exponentially powerful\\
for testing monotonicity of halfspaces\\ \vspace{0.2cm}}

\author{
Xi Chen\thanks{Columbia University, email: \texttt{xichen@cs.columbia.edu}.} 
\and
Rocco A. Servedio\thanks{Columbia University, email: \texttt{rocco@cs.columbia.edu}.}
\and
Li-Yang Tan\thanks{Toyota Technological Institute, email: \texttt{liyang@cs.columbia.edu}.} 
\and
Erik Waingarten\thanks{Columbia University, email: \texttt{eaw@cs.columbia.edu}.}}
\begin{document}
\maketitle
\thispagestyle{empty}
\begin{abstract}
We give a $\poly(\log n, 1/\eps)$-query adaptive algorithm for testing whether 
  an unknown Boolean function $f: \{-1,1\}^n \to \{-1,1\}$, which is promised
to be a halfspace, is monotone versus $\eps$-far from monotone.  Since non-adaptive algorithms are known to 
require almost $\Omega(n^{1/2})$ queries to test whether an unknown halfspace is monotone versus far from monotone, this shows that
adaptivity enables an exponential improvement in the query complexity of monotonicity testing for halfspaces.
\end{abstract}
\newpage 

\hypersetup{linkcolor=magenta}
\hypersetup{linktocpage}


\newpage

\setcounter{page}{1}


\section{Introduction}

Monotonicity testing has been a 
touchstone problem in property testing for more than fifteen years \cite{DGL+99, GGL+00, EKK+00, FLN+02, Fis04, BKR04, ACCL07,HK08, RS09c, BBM12, BCGM12, RRSW12,CS13a, CS13b, CS13c, BRY13,CST144,KMS15,CDST15,BB15}, with many exciting recent developments leading to a greatly improved understanding of the problem in just the past few years.  The seminal work of \cite{GGL+00} introduced the problem and gave an $O(n/\eps)$-query algorithm that tests whether an unknown and arbitrary function
$f\colon \{-1,1\}^n \to \{-1,1\}$ is monotone versus $\eps$-far from every monotone function.  While steady progress followed for non-Boolean functions and for functions over other domains, 
the first improved algorithm for Boolean-valued functions over $\{-1,1\}^n$ was only achieved in \cite{CS13a}, who gave a $\tilde{O}(n^{7/8})\cdot \poly(1/\eps)$-query non-adaptive testing algorithm.  A slightly improved $\tilde{O}(n^{5/6})\cdot \poly(1/\eps)$-query non-adaptive algorithm was given by \cite{CST144}, and subsequently \cite{KMS15} gave a $\tilde{O}(n^{1/2})\cdot \poly(1/\eps)$-query non-adaptive algorithm.

On the lower bounds side, the fundamental class of \emph{halfspaces}\ignore{\green{(which are important objects of study in diverse contexts including circuit complexity, machine learning, and social choice theory)}} has played a major role in non-adaptive lower bounds for monotonicity testing to date.  We discuss lower bounds for two-sided error monotonicity testing of Boolean-valued functions over $\{-1,1\}^n$, and refer the reader to the above references for lower bounds on other variants of the monotonicity testing problem. The first (two-sided) lower bound  was established by Fischer et al \cite{FLN+02}, who used a slight variant of the majority function to give an $\Omega(\log n)$ lower bound for non-adaptive monotonicity testing. More recently, the lower bound of \cite{CDST15}, strengthening \cite{CST144}, shows that for any constant $\delta>0$, there is a constant $\eps=\eps(\delta) > 0$ such that $\Omega(n^{1/2 - \delta})$ non-adaptive queries are required to distinguish \ignore{(with a constant advantage over random guessing)} whether {a Boolean function $f$} --- which is promised to be a halfspace --- is monotone or $\eps$-far from every monotone function. \ignore{In fact this lower bound holds even if $f$ is promised to be a halfspace (since the hard yes- and no-distributions used in \cite{CDST15} are supported on halfspaces only).}Together with the $\tilde{O}(n^{1/2}) \cdot \poly(1/\eps)$-query\ignore{\footnote{We use the tilde to hide a polylogarithmic factor.}} non-adaptive monotonicity testing algorithm of \cite{KMS15}, this shows that halfspaces are ``as hard as the hardest functions'' to non-adaptively test for monotonicity.  Halfspaces are also commonly referred to as ``linear threshold functions'' or LTFs; for brevity we shall subsequently refer to them as LTFs.\vspace{-0.25cm}

\paragraph{The role of adaptivity.}  While the above results largely settle the query complexity of non-adaptive monotonicity testing, the situation is less clear when adaptive algorithms are allowed.  More generally, the power of adaptivity in property testing is not yet well understood, despite being a natural and important question.\footnote{For monotonicity
testing of functions $f\colon[n]^2\rightarrow \{0,1\}$, Berman et al. \cite{BRY:14b} showed that
adaptive algorithms are strictly more powerful than non-adaptive ones (by a factor
of $\log 1/\eps$). For unateness testing of real-valued functions $f \colon \{0, 1\}^n \to \R$, a natural generalization of monotonicity, \cite{BCPRS17} showed that adaptivity helps by a logarithmic factor. We remark that for another touchstone class in property testing, the class of Boolean juntas, it was only very recently shown \cite{STW15, CSTWX17} that adaptive algorithms are strictly more powerful than non-adaptive algorithms.} A recent breakthrough result of Belovs and Blais
\cite{BB15} gives a  $\tilde{\Omega}(n^{1/4})$ lower bound on the query complexity 
  of adaptive
algorithms that test whether $f\colon\{-1,1\}^n\rightarrow \{-1,1\}$ is monotone versus $\eps$-far from monotone, for some absolute constant $\eps > 0$. This result was then improved by \cite{CWX17} to $\tilde{\Omega}(n^{1/3})$. \cite{BB15} also shows that when $f$ is promised to be an ``extremely regular'' LTF, with regularity parameter at most $O(1)/\sqrt{n}$, then $\log n+O_\eps(1)$
adaptive queries suffice.  (We define the ``regularity''
of an LTF in part $(a)$ of Definition \ref{def:non-mon-LTF} below.  Here we note only that every $n$-variable LTF has regularity between $1/\sqrt{n}$ and 1, so $O(1)/\sqrt{n}$-regular LTFs are ``extremely regular'' LTFs.)

A very compelling question is whether adaptivity helps for monotonicity testing of Boolean functions:  can adaptive algorithms go below the \cite{CDST15} $\Omega(n^{1/2 - \delta})$-query lower bound for non-adaptive algorithms?  While we do not know the answer to this question for general Boolean functions\footnote{For very 
special functions such as truncated anti-dictators, it is known \cite{FLN+02} that adaptive algorithms 
are known to be much more efficient than nonadaptive algorithms 
($O(\log n)$ versus $\Omega(\sqrt{n})$ queries) in finding a violation to monotonicity.}, in this work we give a strong positive answer in the case of LTFs, generalizing the upper bound of \cite{BB15} from ``extremely regular'' LTFs to arbitrary unrestricted LTFs.  The main result of this work is an adaptive algorithm with one-sided error
that can test any LTF for 
monotonicity using $\poly(\log n,1/\eps)$ queries:\ignore{\footnote{Theorem \ref{thm:main} was obtained independently of the results of \cite{BB15}.}}

\begin{theorem}[Main] 
\label{thm:main}
There is a $\poly(\log n,1/\eps)$-query\footnote{See Theorem \ref{thm:query-complexity} of Section \ref{sec:detailed} for a detailed description of the algorithm's query complexity; we have made no effort to optimize the particular polynomial dependence on $\log n$ and $1/\eps$ that the algorithm achieves.} adaptive algorithm 
  with the following property: given
$\eps>0$ and black-box access
to an unknown LTF $f\colon \{-1,1\}^n \to \{-1,1\}$, 
\begin{itemize}
\item If $f$ is monotone then 
  the algorithm outputs ``monotone'' with probability $1$;
\item If $f$ is $\eps$-far from every monotone function then
  the algorithm outputs ``non-monotone''\\  with
  probability at least $2/3$. 
  \end{itemize}
\end{theorem}
Recalling that the $\Omega(n^{1/2-\delta})$ non-adaptive lower bound from~\cite{CDST15} is proved using LTFs as both the yes- and no- functions, Theorem~\ref{thm:main} shows that adaptive algorithms are exponentially more powerful than non-adaptive algorithms for testing  monotonicity of LTFs.  Together with the $\tilde{\Omega}(n^{1/3})$ 
  adaptive lower bound from \cite{CWX17},  it also shows that LTFs are exponentially easier to test for monotonicity  
  than general Boolean functions using adaptive algorithms.  \ignore{We hope that the algorithmic techniques we introduced for adaptive monotonicity testing of halfspaces may possibly be useful for developing adaptive monotonicity testing algorithms for more general Boolean functions.}
\subsection{A very high-level overview of the algorithm} \label{sec:hi-level}

The adaptive algorithm of \cite{BB15} for testing monotonicity of ``extremely regular'' LTFs is essentially based on a simple binary search over the hypercube $\{-1,1\}^n$ to find an 
  anti-monotone edge\hspace{0.02cm}\footnote{A 
  \emph{bi-chromatic} edge of $f\colon\{-1,1\}^n\rightarrow \{-1,1\}$
  is a pair $(x,y)$ of points  such that $x,y\in \{-1,1\}^n$ differ at exactly
  one coordinate and satisfy $f(x)\ne f(y)$.
An \emph{anti-monotone} edge of $f$ is a 
  bi-chromatic edge $(x,y)$ that also satisfies $x_i=-1,y_i=1$ for 
  some $i\in [n]$ and $f(x)=1,f(y)=-1$.}.
\cite{BB15} succeeds in analyzing such an algorithm,
  taking advantage of some of the nice structural properties of regular LTFs,
  but it is not clear how to carry out such an analysis for general LTFs.

To deal with general LTFs, our algorithm is more involved and employs an iterative stage-wise approach, running for up to $O(\log n)$ stages.
Entering the $(t+1)$-th stage, the algorithm maintains a restriction $\rho^{(t)}$ that fixes some of the input variables to $f$, and in the $(t+1)$-th stage the algorithm queries $f_{\rho^{(t)}}$, where we write $f_{\rho^{(t)}}$ to denote the 
  function $f$ after the restriction $\smash{\rho^{(t)}}$. 
At a very high level, in the $(t+1)$-th stage the algorithm either

\begin{flushleft}\begin{enumerate}
\item [(i)] Obtains definitive evidence (in the form of an anti-monotone edge) that $f_{\rho^{(t)}}$, and hence $f$, is not monotone.  In this case the algorithm halts and outputs ``non-monotone.'' Or, it\vspace{-0.06cm}

\item [(ii)] Extends the restriction $\rho^{(t)}$ to obtain $\rho^{(t+1)}$.  This is done by fixing a random subset of the variables of expected density 1/2 that are not fixed under $\rho^{(t)}$, and possibly some additional variables, in such a way as to maintain an invariant described later. Or, it\vspace{-0.06cm}

\item [(iii)] Fails to achieve (i) or (ii), which we show is very unlikely to happen.
In this case the algorithm simply halts and outputs ``monotone.''
\end{enumerate}\end{flushleft}

We describe the invariant of $\rho^{(t)}$ maintained in Case (ii) in Section \ref{moredetails}.
One of its implications in particular is that $\smash{f_{\rho^{(t)}}}$ is $\eps'$-far from monotone, where
  $\eps'$ has a polynomial dependence on $\eps$.
As a result, when the number of surviving variables under $\rho^{(t^*)}$ at the beginning of a 
  stage $t^*$ is~at~most $\poly(\log n)$, the algorithm can run the simple ``edge tester'' of~\cite{GGL+00} on $f_{\rho^{(t^*)}}$
  to find an anti-monotone edge with high probability. 
Although the ``edge tester'' has query complexity linear in the number of variables,
  this is affordable since $f_{\rho^{(t^*)}}$ only has $\poly(\log n)$ many variables left.
Case (ii) ensures that there are at most $O(\log n)$  stages overall.  We will also see that each stage makes at most $\poly(\log n, 1/\eps)$ queries;
  hence the overall query complexity is $\poly(\log n, 1/\eps)$.

\subsection{A more detailed overview of the algorithm and why it works}\label{moredetails}

In this section we give a more detailed overview of the algorithm and a high-level sketch of its analysis.
The algorithm only outputs ``non-monotone'' if it identifies an anti-monotone edge,
so~it~will correctly output ``monotone'' on every monotone $f$ with probability 1.  Hence, establishing correctness of the algorithm amounts to showing that if $f$ is an LTF that is $\eps$-far from monotone, then with high probability the algorithm will output ``non-monotone'' when it runs on $f$.  Thus, for the remainder of this section, $f(x) = \sign(w_1 x_1 + \cdots + w_n x_n - \theta)$ should be viewed as being an LTF that is $\eps$-far from monotone. 

A crucial notion for understanding the algorithm is that of a \emph{$(\tau,\gamma,\lambda)$-non-monotone LTF}.

\begin{definition} \label{def:non-mon-LTF}
Given an LTF $f\hspace{-0.03cm}:\hspace{-0.03cm}\{-1,1\}^S\hspace{-0.03cm} \to \{-1,1\}$ of the form $f(x)=\sign(w\cdot x-\theta)$ over~a~set of variables $S$,
  we say it is a $(\tau,\gamma,\lambda)$-\emph{non-monotone  
  LTF with respect to the weights $w$} if it satisfies the following three properties:\vspace{0.08cm} 
\begin{flushleft}\begin{enumerate}
\item [$(a)$] $f$ is \emph{$\tau$-weight-regular}\hspace{0.03cm}\footnote{Our terminology ``weight-regular'' means the same thing as \cite{BB15}'s ``regular.''  We use the terminology ``weight-regular'' to distinguish it from the different notion of ``Fourier-regularity'' which we also require, see 
Section \ref{sec:fourier-regularity}.}
 with respect to $w$, i.e.,
\[\mathop{\max}_{i \in S} |w_i| \leq \tau \cdot \sqrt{\sum_{j \in S} w_j^2} \hspace{0.08cm};\]
\item [$(b)$] $f$ is \emph{$\gamma$-balanced}, i.e.,
$
\left|\hspace{0.03cm}\Ex_{\bx\in \{-1,1\}^n} [f(\bx)]
\hspace{0.01cm}\right|\le 1-\gamma \hspace{0.06cm} ; 
$ and\vspace{-0.12cm}
\item [$(c)$] $f$ has \emph{$\lambda$-significant 
squared negative weights in $w$}, i.e.,
$$
\dfrac{\sum_{i\in S: w_i<0} (w_i)^2}{\sum_{i\in S} (w_i)^2}\ge \lambda.
$$
\end{enumerate}\end{flushleft}
\end{definition}

Looking ahead, an insight that underlies this definition (as well as our algorithm) is that, when $f=\sign(w \cdot x - \theta)$ is a weight-regular LTF that is far from monotone, $f$ must satisfy $(c)$ above~for some large value of $\lambda$ (see Lemma \ref{lem:reg-neg-coeff} for a precise formulation).  
The converse also holds, i.e., an LTF that satisfies all three conditions above must
  be $\eps$-far from monotone for some large value~of~$\eps$ (see Lemma \ref{lem:lem3.1-conv}).
This is indeed the reason why we call such functions $(\tau,\gamma,\lambda)$-\emph{non-monotone}~LTFs.
An additional motivation for the regularity condition $(a)$ is that,
when $f$ satisfies $(c)$ for some value $\lambda\gg \tau$ (the parameter in (a)), a random restriction $\rho$
  (that randomly fixes~half of the variables to uniform values from $\{-1,1\}$) would
  have $f_\rho$ still satisfy (c) with essentially the same $\lambda$.
The balance condition $(b)$, on the other hand, may be viewed as a technical condition that makes it possible for our various subroutines to work efficiently and correctly; we note that if $f$ is not $\gamma$-balanced, then $f$ is trivially $(\gamma/2)$-close to either the monotone function $1$ or the monotone function $-1$.

With Definition \ref{def:non-mon-LTF} in hand, we proceed to a more detailed overview of the algorithm  (still at a rather conceptual level).  The algorithm takes as input black-box access to $f\colon \{-1,1\}^n \to \{-1,1\}$ and a parameter $\eps > 0$.
We remind the reader that in the subsequent discussion $f$ 
  should be~viewed as an $\eps$-far-from-monotone LTF.
For the analysis of the algorithm, we also assume that $f$ takes
  the form of $f(x)=\sign(w_1x_1+\cdots+w_nx_n-\theta)$, for some
  unknown (but fixed\hspace{0.05cm}\footnote{Note that
  $(w,\theta)$ is not unique for a given $f$. 
Here we pick any such pair and stick to it throughout the analysis.}) weight vector $w$ and threshold $\theta$.
They are unknown to the algorithm and will be used in the analysis only.

Our algorithm has two main phases:  first an \emph{initialization} phase, and then the phase consisting of the \emph{main procedure}.   

\medskip

\noindent {\bf Initialization.}  The algorithm runs an initialization procedure called {\tt Regularize-and-Balance}. Roughly speaking, it with high probability either identifies $\smash{f}$ as a non-monotone LTF by finding~an
  anti-monotone edge and halts, or  
constructs a restriction $\smash{\rho^{(0)}}$ such that $\smash{f_{\rho^{(0)}}}$ becomes~a~$\smash{(\tau,\gamma,\lambda_0)}$-non-monotone LTF for suitable parameters $\tau,\gamma,\lambda_0$,
with $\tau= \text{poly}(1/\log n, \eps)$, $\gamma=\eps$, 
  $\lambda_0= { \poly( \eps)}$ and $\tau\ll \lambda_0$.
  In the latter case the algorithm continues with $f_{\rho^{(0)}}$.

\medskip

\noindent {\bf Main Procedure.}  As sketched earlier in Section
\ref{sec:hi-level} the main procedure  operates in a sequence~of $O(\log n)$ stages.  In its $(t+1)$th
stage, it operates on the restricted function $f_{\rho^{(t)}}$  which is assumed to be a $(\tau,\gamma,\lambda_t)$-non-monotone LTF,   and with high probability
  either identifies $f$ as non-monotone and halts, or
  constructs an extension $\rho^{(t+1)}$ of the restriction $\rho^{(t)}$  such that 
  $f_{\rho^{(t+1)}}$ remains $(\tau,\gamma,\lambda_{t+1})$-non-monotone (for some parameter
   $\lambda_{t+1}$ that is only slightly smaller than $\lambda_t$) while
the number of free variables in $\rho^{(t+1)}$ drops by a constant factor.

\def\stars{\textsc{stars}}

To describe each stage in more detail, we need the following notation
  for restrictions.
Given a restriction $\smash{\rho \in \{-1, 1, *\}^{[n]}}$, 
  we use $\stars(\rho)$ to denote the set of indices 
  that are not fixed in $\rho$, i.e., the set of $i$ such that $\rho(i)=*$. 
  Given $f\colon \{-1, 1\}^n \rightarrow \{-1, 1\}$ of the form $f(x) = \sign(\sum w_i x_i - \theta)$, we let $\smash{f_{\rho} \colon \{-1, 1\}^{\stars(\rho)} \rightarrow \{-1, 1\}}$ denote the function $f$ after the restriction $\rho$: 
\begin{eqnarray*}
&f_{\rho}(x) = \sign\left(\sum_{i\in\stars(\rho)}w_i\cdot x_i + \sum_{j \notin\stars(\rho)} w_j\cdot \rho(j) - \theta\right).&
\end{eqnarray*} 
We stress than the weights of $f_{\rho}$ remain $w_i$ while the threshold is $\smash{\theta-\sum_{j\notin\stars(\rho)}w_j\cdot\rho(j)}$.

Now for the $(t+1)$th stage, where $t= 0, 1, 2,\dots,$ the main procedure carries out the following sequence of steps (we defer discussion of how these steps are implemented to Section \ref{sec:detailed}). 
Below for convenience we let $g$ denote $f_{\rho^{(t)}}$,
  the function that the algorithm operates on in the $(t+1)$th stage. 

\begin{flushleft}\begin{enumerate}
\item Draw a random subset $A_t \subset \stars(\rho^{(t)})$, which consists of
  roughly half of its variables. 
  Assuming that $\tau\ll \lambda_t$, we have that, with high probability,
  $A_t$ partitions the positive and negative weights {roughly} evenly 
  and {the collection of} weights of variables in $\stars(\rho^{(t)})\setminus A_t$
   {has} $\lambda_{t+1}$-significant squared negative weights
  for some $\lambda_{t+1}$ that is only slightly smaller than $\lambda_t$.
(This also justifies the assumption of $\tau\ll \lambda_t$ at the beginning.)\vspace{-0.05cm}

\item Find a restriction $\smash{\rho' \in \{-1,1,\ast\}^{\stars(\rho^{(t)})}}$ that fixes the variables in $A_t$   {in} such {a way} that
$g_{\rho'}$ is $0.96$-balanced.
The exact constant $0.96$ here is not important as long as it is 
  close enough\\ to $1$. 
Note that $g_{\rho'}$ is more balanced than $g$ {is promised to be} 
  (i.e., $(\gamma=\eps)$-balanced and we may assume that
  $\eps\le 0.5$). This helps in the last step of the stage.
Our analysis shows that if $g$ is $(\tau,\gamma,\lambda_t)$-non-monotone,
  then this step succeeds with high probability.\vspace{-0.05cm}

\item Find a set $H_t \subset \stars(\rho^{(t)}) \setminus A_t$ that contains those variables $x_i$ that have ``high influence'' in $g_{\rho'}$. Intuitively,
  $H_t$ contains variables of $g_{\rho'}$ that violate
  the $\tau$-weight-regularity condition; after its removal,
    {the collection of} weights of variables in $\stars(\rho^{(t)})\setminus (A_t\cup H_t)$ 
    {becomes} $\tau$-weight-regular again.\vspace{-0.05cm}

\item For each $i \in H_t$, find a bi-chromatic edge of $g_{\rho'}$ on the $i$th coordinate (this can be done efficiently because the variables in $H_t$ all have high influence in $g_{\rho'}$), which reveals the sign of $w_i$.
If an anti-monotone edge is found, halt and output ``non-monotone;'' 
  otherwise, we know that the weight of every variable in $H_t$
  is positive.\vspace{-0.05cm}
      
\item Finally, find a restriction $\smash{\rho'' \in \{-1,1,\ast\}^{\stars(\rho^{(t)})}}$, which extends $\rho'$ and fixes the variables in $A_t \cup H_t$, such that 
$g_{\rho''}$ is $\gamma$-balanced. Our analysis shows that if $g$ is
  $(\tau,\gamma,\lambda_t)$-non-monotone and $g_{\rho'}$
  is $0.96$-balanced, then this step succeeds with high probability. 
By Step 3, $g_{\rho''}$ is $\tau$-weight-regular.
In addition, $g_{\rho''}$ has 
  $\lambda_{t+1}$-significant squared negative weights because
  of Step 1 and Step 4 (which makes sure that all variables in $H_t$
  have positive weights). 
At the end, we set $\rho^{(t+1)}$ to be the composition of 
  $\rho^{(t)}$ and $\rho''$ and move on to the next stage.
\end{enumerate}\end{flushleft}

To summarize,
our analysis shows that if $f_{\rho^{(t)}}$ is $(\tau,\gamma,\lambda_t)$-non-monotone (entering the $(t+1)$th stage) then with high probability the algorithm
  in the $(t+1)$th stage either finds an anti-monotone edge and halts,
  or 
  finds an extension $\rho^{(t+1)}$ of $\rho^{(t)}$ such that
\begin{flushleft}\begin{itemize}
\item[i)]
The new function $f_{\rho^{(t+1)}}$ is $(\tau,\gamma,\lambda_{t+1})$-non-monotone (entering the
  $(t+2)$th stage), where the parameter $\lambda_{t+1}$ is only slightly smaller than $\lambda_t$ (more on this below); and \vspace{-0.06cm}
\item[ii)]
The number of surviving variables in $\rho^{(t+1)}$
  is only about half of that of $\rho^{(t)}$.  
\end{itemize}\end{flushleft}
This implies that, with high probability, the main procedure
  within $O(\log n)$ stages either 
  finds~an anti-monotone edge and returns the correct answer ``non-monotone''
  or constructs a restriction~$\rho^{(t)}$ 
  such that $f_{\rho^{(t)}}$ is $(\tau,\gamma,\lambda_{t})$-non-monotone
  and the number of 
  surviving variables under $\rho^{(t)}$ is at~most 
  $m = \poly(\log n, 1/\eps)$.
For the latter case,    
  our analysis (Lemma \ref{lem:lem3.1-conv}) 
  together with the fact that $\lambda_t$ drops only slightly in each stage
  show that $f_{\rho^{(t)}}$ remains $\eps' =\poly(\eps)$-far from monotone.
Thus, the algorithm concludes by running the ``edge tester'' from \cite{GGL+00} to $\eps'$-test the $m$-variable function $f_{\rho^{(t)}}$, which uses
  $O(m/\eps')=\poly(\log n,1/\eps)$ queries to $f_{\rho^{(t)}}$
  and finds an anti-monotone edge with high probability.
To summarize, when $f$ is an LTF that is $\smash{\eps}$-far from monotone,
  our~algorithm finds an anti-monotone edge 
  and outputs ``non-monotone'' with high probability.  
As discussed earlier at the beginning of Section \ref{moredetails} about its one-sideness,
  the correctness of the algorithm follows.

\subsection{Relation to previous work}

We have already discussed how our main result, Theorem \ref{thm:main}, relates to the recent upper and lower bounds of
\cite{KMS15,CDST15,BB15} for monotonicity testing.  At the level of techniques, several aspects of our algorithm are reminiscent of some earlier work in property testing of Boolean functions and probability distributions as we describe below.

At a high level, the $\poly(1/\eps)$-query algorithm of \cite{MORS:10} for testing whether a function is~an LTF identifies high-influence variables and ``deals with them separately'' from other variables, as does our algorithm.  The more recent algorithm of \cite{RonServedio15}, for testing whether a function is a signed majority function, like our algorithm proceeds in a series of stages which successively builds up~a restriction by fixing more and more variables.  Like our algorithm the \cite{RonServedio15} algorithm makes only $\smash{\poly(\log n, 1/\eps)}$ adaptive queries, but there are many differences both between the two algorithms and between their analyses. To briefly note a few of these differences, the \cite{RonServedio15} algorithm has~two-sided error while our algorithm has one-sided error; the former also heavily leverages both the very ``rigid'' structure of the degree-1 Fourier coefficients of any signed majority function and the near-perfect balancedness of any signed majority function between the two outputs 1 and $-1$, neither of which hold in our setting. Finally, we note that the general approach of iteratively selecting {and retaining} a random subset of the remaining ``live'' elements, then doing some additional pruning to identify, check, and discard a small number of ``heavy'' elements, then proceeding to the next stage is reminiscent of the {\sc Approx-Eval-Simulator} procedure of \cite{CRS15}, which deals with testing probability distributions in the ``conditional sampling'' model.

\subsection{Organization} 

In Section~\ref{sec:background} we recall the necessary background concerning monotonicity, LTFs, and restrictions, and state a few useful algorithmic and structural results from prior work.  In Section~\ref{sec:new-structural} we establish several new structural results about ``regular'' LTFs: we first show that its distance to monotonicity corresponds (approximately) to its total amount of squared
  negative coefficient weights; 
we also prove that its distance to monotonicity is preserved under a random restriction to a set of its non-decreasing variables. In Section~\ref{sec:new-algorithmic} we present and analyze some simple algorithmic subroutines that will be used to identify high influence variables and check that they are non-decreasing. Finally in Section~\ref{sec:detailed}, we give a detailed description of our overall algorithm
  for  testing monotonicity of LTFs, and prove its correctness, establishing our main result (Theorem~\ref{thm:main}).


\section{Background}
\label{sec:background}

We write $[n]$ for $\{1,\dots,n\}$,
  and use boldface letters (e.g., $\bx$ and $\mathbf{X}$) 
  to denote random variables.

We briefly recall some basic notions.  A function $f\colon \{-1,1\}^n \to \{-1,1\}$ is \emph{monotone} (short for ``monotone non-decreasing'') if $x \preceq y$ implies $f(x) \leq f(y)$, where ``$x \preceq y$'' means that $x_i \leq y_i$ for all $i \in [n].$  A function $f$ is \emph{unate} if there is a bit vector $a \in \{-1,1\}^n$ such that $f(a_1 x_1,\dots,a_n x_n)$ is monotone.  It is well known
that every LTF (defined below) is unate.

We measure distance between functions $f,g\colon \{-1,1\}^n \to \{-1,1\}$ with respect to the uniform distribution, so we say that $f$ and $g$ are \emph{$\eps$-close} if 
\[ \dist(f,g) \coloneqq \mathop{\Pr}_{\bx \in \{-1,1\}^n}\big[f(\bx) \neq g(\bx)\big] \leq \eps, \]
and that $f$ and $g$ are \emph{$\eps$-far} otherwise.  A function $f$ is \emph{$\eps$-far from monotone} if it is $\eps$-far from every monotone function $g.$  We write
$\dist(f, \textsc{Mono})$ to denote the minimum value of $\dist(f,g)$ over all monotone functions $g$.
Throughout the paper all probabilities and expectations are with respect to the uniform distribution over $\{-1,1\}^n$ unless otherwise indicated.
As indicated in Definition \ref{def:non-mon-LTF}, 
  we say that a $\{-1,1\}$-valued function $f$ is 
  \emph{$\gamma$-balanced} if 
  \[ \left|\mathop{\E}_{\bx\in \{-1,1\}^n}[f(\bx)]
  \hspace{0.02cm}\right|\le 1-\gamma.\]

A function $g\colon \{-1,1\}^n \to \{-1,1\}$ is a \emph{junta} {over $S  \subseteq [n]$} if $g$ depends only on the coordinates in $S$.  We say $f$ is \emph{$\eps$-close} to a \emph{junta} over $S$ if $f$ is $\eps$-close to $g$ for some $g$ that is a junta over $S$.

\subsection{LTFs and weight-regularity} \label{sec:weight-regularity}

A function $f\colon \{-1,1\}^n \to \{-1,1\}$ is an \emph{LTF} 
(also commonly referred to as a \emph{halfspace})
   if there exist real weights $w_1,\dots,w_n \in \R$
and a real threshold $\theta \in \R$ such that 
\[
f(x) = 
\begin{cases}
1 & \text{~if~} w_1 x_1 + \cdots + w_n x_n \geq \theta,\\
-1 & \text{~if~} w_1 x_1 + \cdots + w_n x_n < \theta.
\end{cases}
\]
We say that $w=(w_1,\dots,w_n)$ are the \emph{weights} and $\theta$ the \emph{threshold} of
the LTF, and we say that $(w,\theta)$ \emph{represents} the LTF $f$, or simply that $f(x)$ is the LTF
given by $\sign(w \cdot x - \theta).$
Note that for any LTF $f$ there are in fact infinitely many pairs $(w,\theta)$ that represent $f$; 
we fix a particular pair $(w,\theta)$ for each $n$-variable LTF $f$ and work with it in what follows.

An important notion in our arguments is that of \emph{weight-regularity}. As indicated in Definition \ref{def:non-mon-LTF}, given a weight vector $w \in \R^n$, we say that $w$ is \emph{$\tau$-weight-regular} if no more than a
$\tau$-fraction of the $2$-norm of $w=(w_1,\dots,w_n)$ comes from any single coefficient $w_i,$ i.e.,
\begin{equation} \label{eq:weight-regular}
\max_{i \in [n]} |w_i| \leq \tau \cdot \sqrt{w_1^2 + \cdots + w_n^2}.
\end{equation}
If we have fixed a representation $(w,\theta)$ for $f$ such that $w$ is $\tau$-weight-regular, we frequently abuse the terminology and say that $f$ is $\tau$-weight-regular.

\subsection{Fourier analysis of Boolean functions and Fourier-regularity}
\label{sec:fourier-regularity}

Given a function $f \colon \bn \to \R$, we define its \emph{Fourier
coefficients} by $\hat{f}(S) = \E [f\cdot x_S]$ for each $S\subseteq [n]$, where $x_S$
  denotes $\prod_{i\in S} x_i$, and we have that
$\smash{f(x) = \sum_S \hat{f}(S)\cdot x_S}$. We will be particularly interested
in $f$'s \emph{degree-$1$} coefficients, i.e., $\smash{\hat{f}(S)}$ for $|S|
= 1$; we will write these as $\smash{\hat{f}(i)}$ rather than
$\hat{f}(\{i\})$. \hspace{-0.05cm}We recall \emph{Plancherel's identity} $\la
f, g \ra = \sum_S \hat{f}(S) \hat{g}(S)$, which has as a special
case \emph{Parseval's identity}, $\smash{\E_{\bx}[f(\bx)^2] = \sum_S
\hat{f}(S)^2}$. It follows that every $f \colon \{-1,1\}^n \to
\{-1,1\}$ has $\sum_S \smash{\hat{f}(S)^2 = 1}$.

We further recall that, for any unate function $f\colon \{-1,1\}^n \to \{-1,1\}$
(and hence any LTF), we have $\smash{|\hat{f}(i)| = \Inf_i(f)}$, where the \emph{influence} of variable $i$ on $f$ is 
\[ \Inf_i(f)=\mathop{\Pr}_{\bx\in
  \{-1,1\}^n}\big[f(\bx) \neq f(\bx^{\oplus i})\big],\]
where $x^{\oplus i}$ is the vector obtained from $x$ by flipping coordinate $i$.

We say that $f\colon \{-1,1\}^n \to \{-1,1\}$ is \emph{$\tau$-Fourier-regular} if $\max_{i \in [n]}|\hat{f}(i)| \leq \tau$.  Section~\ref{sec:useful-structural} below summarizes some useful relationships between weight-regularity  and Fourier-regularity of LTFs.

\subsection{Restrictions}

 A \emph{restriction $\rho$}  is an element of $\{-1,1,\ast\}^{[n]}$; we view $\rho$ as a 
partial assignment to the $n$ variables $x_1,\dots,x_n$, where $\rho(i)=\ast$ indicates that variable $x_i$ is unassigned. 
We write $\supp(\rho)$ to denote the set of indices $i$ such that $\rho(i)\in \{-1,1\}$ and 
  $\stars(\rho)$ to denote the set of $i$ such that $\rho(i)=*$
  (and thus, $\stars(\rho)$ is the complement of $\supp(\rho)$). 
 
 Given restrictions $\rho,\rho' \in \{-1,1,\ast\}^{[n]}$ we say that $\rho'$ is an \emph{extension} of $\rho$ if $\supp(\rho) \subseteq \supp(\rho')$ and 
$\rho'(i)=\rho(i)$ for all $i\in \supp(\rho)$. 
 If $\rho$ and $\rho'$ are restrictions with disjoint support we write $\rho \rho'$ to denote the \emph{composition} of these two restrictions (that has support $\supp(\rho) \cup \supp(\rho')$).

\subsection{Useful algorithmic tools from prior work}

We recall some algorithmic tools for working with black-box functions $f\colon\{-1,1\}^n \to \{-1,1\}$.

\medskip

\noindent {\bf Estimating sums of squares of degree-1 Fourier coefficients.} We first recall Corollary~16~of \cite{MORS:10} (slightly specialized to our context):

\begin{lemma}
 [Corollary~16 \cite{MORS:10}]
 \label{lem:estdeg1} There is a procedure {\tt Estimate-Sum-of-Squares}$(f,T,\eta,\delta)$ with the following properties.  Given
as input black-box access to $f\colon \{-1,1\}^n \rightarrow \{-1,1\}$,  a subset $T \subseteq [n]$,
and parameters $\eta,\delta>0$, it runs in time $O(n\cdot\log(1/\delta)/\eta^4)$,
makes $O(\log(1/\delta)/\eta^4)$ queries, and with probability at least $1-\delta$ outputs an 
estimate of $\sum_{i \in T} \hat{f}(i)^2$ that is accurate to within an additive $\pm
\eta$.
\end{lemma}
 
\noindent {\bf Checking Fourier regularity.} We recall Lemma~18 of \cite{MORS:10}, which is an easy consequence of Lemma \ref{lem:estdeg1}:

\red{\begin{lemma}  [Lemma~18 \cite{MORS:10}] \label{lem:testnonregular}
There is a procedure
{\tt Check-Fourier-Regular}$(f,T,\tau,\delta)$ with the following properties.
Given as input black-box access to $f\colon \{-1,1\}^n \rightarrow \{-1,1\}$, $T \subseteq [n]$,
and parameters $\tau,\delta>0$, it runs in time $O(n\cdot \log(1/\delta)/\tau^{16})$, makes $O(\log(1/\delta)/\tau^{16})$ 
queries, and
\begin{itemize}
\item If $|\hat{f}(i)| \geq \tau$ for some
$i \in T$ then it outputs ``not regular''  with probability $1 -
\delta$;\vspace{-0.1cm}

\item If every $i \in T$ has $|\hat{f}(i)| \le \tau^2/4$ then it
outputs ``regular''  with probability $1 -
\delta$.
\end{itemize}
\end{lemma}
} 

\noindent {\bf Estimating the mean.}  For completeness we recall the following simple fact (which follows from a standard 
Chernoff bound):

\begin{fact} \label{fact:estimate-mean}
There is a procedure {\tt Estimate-Mean}$(f,\eps,\delta)$ with the following properties.
Given as~input black-box access to $f\colon\{-1,1\}^n\rightarrow \{-1,1\}$ and 
  $\eps,\delta>0$, it makes $O(\log(1/\delta)/\eps^2)$ 
 queries and with probability at least $1-\delta$ it outputs a value $\tilde{\mu}$ such that
$|\tilde{\mu}-\mu| \leq \eps$, where $\mu = \E_{\bx \in \{-1,1\}^n}[f(\bx)].$
\end{fact}
 
\noindent {\bf The edge tester of \cite{GGL+00}.}   We recall the performance guarantee of the ``edge tester'' (which works by querying both endpoints of uniform random edges and outputting ``non-monotone'' if and only if it encounters an anti-monotone edge):

\begin{theorem}[\cite{GGL+00}] \label{fact:edge-tester}
There is a procedure {\tt Edge-Tester}$(f,\eps,\delta)$ with the following properties:  Given black-box access to $f\colon \{-1,1\}^n \to \{-1,1\}$ and parameters $\eps,\delta>0$, it makes $O({n} \log(1/\delta)/\eps)$ queries and outputs either ``monotone'' or ``non-monotone'' such that
\begin{itemize}
\item If $f$ is monotone then it outputs ``monotone'' with probability 1;\vspace{-0.1cm}
\item If $f$ is $\eps$-far from monotone then it outputs ``non-monotone'' with probability at least $1-\delta.$
\end{itemize}
\end{theorem}

\subsection{Useful structural results from prior work} \label{sec:useful-structural}

\noindent {\bf Gaussian distributions and the Berry--Ess\'een theorem.}  
Recall that the p.d.f. of the standard Gaussian distribution $\calN(0,1)$ with mean 0 and variance 1 is given by
\[
\phi(x) = {\frac 1 {\sqrt{2 \pi}}}\cdot e^{-x^2/2}.
\]
It is useful to have upper and lower bounds on Gaussian tails.  We have the following (see Example 2.1 and Equation 2.51 of \cite{TAILBOUND} for the upper and lower bounds, respectively):  for $t > 0,$ we have
\begin{equation}\label{gatail}
\left({\frac 1 t} - {\frac 1 {t^3}}\right)\cdot  
\frac{1}{\sqrt{2\pi}}\cdot e^{-t^2/2} 
\leq \mathop{\Pr}_{\bz \sim \calN(0,1)}[\bz \geq t] \leq e^{-t^2/2}.
\end{equation}
We also need the following Gaussian anti-concentration bound.

\begin{fact}[Gaussian anti-concentration]
Let $\bz$ be a random variable drawn from a Gaussian distri\-bution with variance $\sigma^2$. Then for all $\eps > 0$, we have $\sup_{\theta \in \R} \{ \hspace{0.03cm}\Pr[\hspace{0.03cm}|\bz - \theta| \leq \eps \sigma 
\hspace{0.03cm}]\hspace{0.03cm}\} \leq \eps$.
\end{fact}

The Berry--Ess\'een theorem (see e.g., \cite{Feller})
is a version of the central limit theorem for sums of independent random variables (stating that such a sum converges to a normal distribution) that provides a quantitative error bound.  It is useful for analyzing weight-regular LTFs and we recall it below (as well as the standard 
  Hoeffding inequality).
 \begin{theorem} [Berry--Ess\'een] \label{BerryEsseen}
    \label{thm:be} Let $\ell(\bx) = c_1 \bx_1 + \cdots + c_n \bx_n$ be a
    linear form of $n$ unbiased, independent random 
    $\{\pm 1\}$-valued variables $\bx_i$.  Let $\tau$ be such that
    $|c_i| \leq \tau$ for all $i$, and let
    $\sigma = (\sum c_i^2)^{1/2}$.  Write $F$ for the c.d.f.\ of
    $\ell(\bx)/\sigma$, i.e., $F(t) = \Pr\hspace{0.03cm}[
    \hspace{0.03cm}\ell(\bx)/\sigma \leq t\hspace{0.03cm}]$.  Then
    for all $t \in \R$, we have that
    $
    \left| F(t) - \Phi(t) \right| \leq \tau/\sigma,
    $
    where $\Phi$ denotes the c.d.f.\ of a standard $\calN(0,1)$ Gaussian random variable.
 \end{theorem}

\begin{theorem}[Hoeffding's Inequality]
Let $\bx$ be a random variable drawn uniformly from $\bn$. Let $w \in \R^d$ and $t > 0$. Then we have
\[ \mathop{\Pr}_{\bx}\big[|\bx \cdot w| \geq t\big] \leq 2\exp\left(-\frac{t^2}{2\| w\|_2^2}\right)\quad\text{and}\quad \mathop{\Pr}_{\bx}\big[ \bx \cdot w \geq t\big] \leq  \exp\left(-\frac{t^2}{2\| w\|_2^2}\right).\vspace{0.08cm} \]
\end{theorem}

\noindent {\bf Weight-regularity versus Fourier-regularity for LTFs.}
An easy argument, using the Berry--Ess\'een inequality above, shows that weight-regularity always implies Fourier-regularity for LTFs:

\begin{theorem} [Theorem~38 of \cite{MORS:10}] \label{thm:38} 
Let $f\colon \{-1,1\}^n \to \{-1,1\}$ be a $\tau$-weight-regular LTF.  Then $f$ is $O(\tau)$-Fourier-regular.
\end{theorem}

The converse is not always true; for example, the constant 1 function, which is $\tau$-Fourier-regular for all $\tau>0$, 
may be written as $f(x) = \sign(x_1 + 2).$  However, if we additionally impose the condition that $f$ is not too biased towards $+1$ or $-1$, then a converse holds.  Sharpening an earlier result (Theorem 39 of \cite{MORS:10}),  Dzindzalieta \ignore{and G\"{o}tze have}has proved the following:

\begin{theorem}[Theorem 20 of \cite{Dzindzalieta14}]\label{Dzindzalieta}
Let $f(x)=\sign(w\cdot x-\theta)$ be an LTF  
  such that $|\E_{\bx}[f(\bx)]|\leq1-\gamma$.
If $f$ is $\tau$-Fourier-regular, then it is also $O(\tau/\gamma)$-weight-regular.
\end{theorem}

\noindent
{\bf Making LTFs Fourier-regular by fixing high-influence variables.}  Finally, we will need the following simple result (Proposition~62 from
\cite{MORS:10}), which shows that LTFs typically become Fourier-regular when their highest-influence variables are fixed to constants:

\begin{proposition}  \label{prop:its-cool} Let $f \colon \bits^n \to \bits$ be an LTF and let
$J \supseteq \{i : |\hat{f}(i)| \geq \beta\}$.  Then $f_{\rho}$ is not
$(\beta/\eta)$-Fourier-regular for at most an $\eta$-fraction of all $2^{|J|}$
restrictions $\rho$ that fix variables in $J$.
\end{proposition}


\section{New structural results about LTFs}
\label{sec:new-structural} 
Our analysis requires a few new structural results about LTFs.  
We collect and prove these results in
this section. Readers who are eager to proceed to the algorithm and its analysis
  of correctness can skip this section and refer back to it as needed.

\subsection{Far-from-monotone weight-regular LTFs have significant 
   squared negative\\ weights, and vice versa}

The main idea of this subsection is that for weight-regular LTFs, the distance to monotonicity~corresponds (approximately) to its total amount of squared weights
  of negative coefficients (under any representation $(w,\theta)$).
  Lemma \ref{lem:reg-neg-coeff} shows that if $f$ is far from monotone then this quantity is large, and Lemma \ref{lem:lem3.1-conv} establishes a converse (both for weight-regular LTFs). We note that Lemma \ref{lem:reg-neg-coeff}~is essentially equivalent to a lemma proved in \cite{BB15}.

We introduce some notation.
Given an LTF $f\colon\{-1,1\}^n\rightarrow \{-1,1\}$ with   $f(x)=\sign(w\cdot x-\theta)$,  we use $P=P(f)$ and $N=N(f)$ to
   denote the set of non-negative and negative indices, respectively: 
 $P = \{i \in [n] \colon w_i \geq 0\}$ and $N = \{ j\in [n] \colon w_j < 0 \}$. 
We also let $\pos(f)$ and $\negg(f)$ denote the sum of squared weights of positive and negative coefficients, respectively: 
\begin{align*}
\pos(f) = \sum_{i \in P} w_i^2 \quad\ \text{and}\ \quad
\negg(f) = \sum_{j \in N} w_j^2.
\end{align*}
Recall that we say $f$ has $\lambda$-\emph{significant squared negative weights}
  if $\negg(f)/(\pos(f)+\negg(f))\ge \lambda$. 
   
We start with the proof of Lemma \ref{lem:reg-neg-coeff}.   
   
\begin{lemma} \label{lem:reg-neg-coeff}
Let $f\colon\{-1,1\}^n\rightarrow \{-1,1\}$ be an LTF given by $f(x)=\sign(w\cdot x-\theta)$. 
If $f$ is~both $\eps$-far from monotone and $\tau$-weight-regular 
  for some $\tau\leq \eps/16$,  then $f$ must have $\lambda$-significant 
  squared negative weights, where
$
\lambda= {\eps^2}/(16\ln(8/\eps))$. 
\end{lemma}

\begin{proof}  For convenience, we normalize all the weights so that $\pos(f) + \negg(f)=1$. Then it suffices to show that $\negg(f) \geq  {\eps^2}/({16\ln(8/\eps))}.$  Since $f$ is $\tau$-weight-regular, we have
  that $\max_i |w_i|\le \tau$. We  also assume that  $\pos(f) > 1/2$, since otherwise $\negg(f) \geq 1/2$ and we are already done.

Let $g\colon\{-1, 1\}^n \rightarrow \{-1, 1\}$, 
  $g(x) = \sign(\sum_{i \in P} w_i\cdot x_i - \theta)$ be the LTF obtained by removing the negative weights from $f$.  By independence of 
$(\bx_i)_{i \in P}$ and $(\bx_j)_{j \in N}$ for uniform $\bx \in \{-1,1\}^n$,  
\begin{align}
\mathop{\Pr}_{\bx}[f(\bx) = g(\bx)] &\geq \mathop{\Pr}_{\bx}\left[\hspace{0.06cm}\left| \sum_{i \in P} w_i \cdot\bx_i - \theta \right| \geq \frac{\eps\sqrt{\pos(f)}}{2}\hspace{0.06cm}\right] \cdot \mathop{\Pr}_{\bx}\left[\hspace{0.06cm} \left| \sum_{j \in N} w_j\cdot \bx_j \right| \leq \frac{\eps\sqrt{\pos(f)}}{2}\hspace{0.06cm}\right] \label{eq:reg-neg-1}.
\end{align}
We can lower bound the first probability by
\begin{align}
\mathop{\Pr}_{\bx}\left[\hspace{0.06cm}\left| \sum_{i \in P} w_i\cdot \bx_i - \theta \right| \geq \frac{\eps\sqrt{\pos(f)}}{2}\hspace{0.06cm}\right]
                      &\geq \mathop{\Pr}_{\bz \sim \calN(0, \pos(f))}\left[\hspace{0.03cm}|\bz - \theta| \geq \frac{\eps\sqrt{\pos(f)}}{2}\hspace{0.06cm}\right] - \frac{2\tau}{\sqrt{\pos(f)}} \label{eq:be1}\\[0.2ex]
                      &\geq  1 - \frac{\eps}{2} - \frac{2\tau}{\sqrt{\pos(f)}} \geq 1 - \frac{3\eps}{4}, \label{eq:gac}
                      \end{align}
by using the Berry--Ess\'een theorem for (\ref{eq:be1}). Note that even though the error term from the Berry--Ess\'een theorem is $ {\tau}/{\sqrt{\pos(f)}}$, the set
we are interested in is indeed  the union of \emph{two} intervals,~giving us $ {2\tau}/{\sqrt{\pos(f)}}$. In (\ref{eq:gac}) we used Gaussian anti-concentration, $\pos(f)\ge 1/2$, and~$\tau \leq \eps/16$.
On~the other hand, we can lower bound the second probability
  of (\ref{eq:reg-neg-1}) by 
                    \begin{align}
                    \mathop{\Pr}_{\bx}\left[ \hspace{0.06cm}\left| \sum_{j \in N} w_j\cdot \bx_j \right| \leq \frac{\eps\sqrt{\pos(f)}}{2}\hspace{0.06cm}\right]
                    &\geq 1 - 2\cdot\exp\left({-\frac{\eps^2\hspace{0.03cm}\pos(f)}{  8\hspace{0.03cm}\negg(f)}}\right), \label{eq:reg-neg-4}
\end{align}
using Hoeffding's Inequality. 
Moreover, since $g$ is a monotone function and $f$ is $\eps$-far from monotone we have $\Pr_{\bx}[f(\bx)=g(\bx)] \leq 1 - \epsilon$. Combining all these inequalities, we get
\[ 1 - \eps \geq \left( 1 - \frac{3\eps}{4} \right) \left( 1 - 2\cdot\exp\left({-\frac{\eps^2\hspace{0.03cm}\pos(f)}{  8\hspace{0.03cm} \negg(f)}}\right)\right),\]
which, together with $\pos(f) \geq {1/2}$, implies that $\negg(f) \geq  {\eps^2}/({16\ln(8/\eps)})$ as claimed.
\end{proof}

The following lemma, which will be used to prove a converse to Lemma \ref{lem:reg-neg-coeff}, says that if $f$~is~an LTF that is close to a monotone function, then $f$ must be close to the LTF obtained by erasing all its negative weights.  Recall  $\dist(f, \textsc{Mono})$ is the distance from $f$ to the closest monotone function.

\begin{lemma}\label{lem:closest-mon-ltf}
Let $f\colon \{-1,1\}^n \to \{-1,1\}$ be an LTF given by 
  $f(x) = \sign(\sum_{i \in [n]} w_i\cdot x_i  - \theta)$ 
and~let $g$ denote the
(monotone) LTF given by $g(x) = \sign(\sum_{i \in P} w_i\cdot x_i - \theta)
$.
Then $$\dist(f,g) = \dist(f, \textsc{Mono}).$$
\end{lemma}

\red{\begin{proof}
We view an assignment $x \in \{-1, 1\}^n$ as the concatenation of $x' \in \{-1, 1\}^P$ and $y \in \{-1, 1\}^N$, and we write $f(x', y)$ for $f(x)$. We denote
\[ p_{x'} = \Ex_{y} \left[ \frac{1}{2} + \frac{f(x', y)}{2} \right] \]
as the fraction of $1$ inputs. 
It is clear that $g$ depends only on $x'$ (so we may write $g(x', y)$ simply as $g(x')$) and that $g(x')$ is monotone. Additionally, (by symmetry) we have that 
\[ g(x') = 1 \iff p_{x'} \geq \frac{1}{2}. \]
Thus, $\dist(f, g) = \E_{x'}[ \min \{ p_{x'}, 1 - p_{x'} \} ]$. Implicitly in Lemma 3.11 in \cite{KMS15}, it is shown that when $f$ is unate (which is the case for LTFs), $\dist(f, \textsc{Mono}) = \E_{x'} [\min\{ p_{x'}, 1 - p_{x'}\}]$ which finishes the proof.
\ignore{The proof closely follows that of Lemma~3.11 in \cite{KMS15}.  For convenience we view both $f$ and $g$ as 0/1-valued, i.e., the output $-1$ is replaced by 0.
We also view an assignment $x\in \{-1,1\}^n$ as the concatenation of
  $x'\in \{-1,1\}^P$ and $y\in \{-1,1\}^N$ and write
  $f(x',y)$ for $f(x)$.

It is clear that $g$ depends only on $x'$ (so we may write $g(x',y)$ simply as $g(x')$) and that~$g(x')$~is monotone.  For each fixed $x' \in \{-1,1\}^P$ we have that $g(x')$ is a fixed value in $\{0,1\}$ so 
\[ p_{x'} \coloneqq \mathop{\Pr}_{\by\in\{-1,1\}^N}\big[f(x',\by) \neq g(x')\big]= \dist\big(f(x',\cdot),g(x')\big),\]
and $\dist(f,g)=\Ex_{\bx'}[p_{\bx'}].$  Now the fact that 
\begin{eqnarray*}
 \mathop{\Pr}_{\by}\left[\sum_{j \in N}w_j\cdot \by_j \geq 0\right],\ \mathop{\Pr}_{\by}\left[\sum_{j \in N}w_j\cdot \by_j \leq 0\right] \ge 1/2
\end{eqnarray*}
(by symmetry) implies that $ p_{x'} \leq 1/2$ for each $x'\in \{-1,1\}^P$, so $1-p_{x'} \geq 1/2$ and we get
\[
2\Ex_{\bx'}\big[\hspace{-0.03cm}\Var[f(\bx',\cdot)]\big]=
\Ex_{\bx'}\big[2p_{\bx'}(1-p_{\bx'})\big] \geq \Ex_{\bx'}\big[p_{\bx'}\big]=\dist(f,g).
\]
Since $f(x',y)$ is an anti-monotone function of $y$  for each fixed $x'$, Lemma~3.12 of \cite{KMS15}
implies that 
$\Ex_{\bx'}[\Var[f(\bx',\cdot)]]\leq \dist(f, \textsc{Mono})$, and the proof is concluded.}
\end{proof} 
}

Here is the converse to Lemma \ref{lem:reg-neg-coeff}:
\begin{lemma}\label{lem:lem3.1-conv}
Let $f(x) = \sign(\sum_{i \in [n]} w_i\cdot x_i - \theta)$ be $(\tau,\gamma,\lambda)$-non-monotone with $\tau\le \sqrt{\lambda}/16$. 
Then
\[
\dist(f,\textsc{Mono})\ge \min\left\{  \Omega\big({  {\sqrt{\lambda}\hspace{0.03cm}\gamma^2} }\big) - O(\tau),\hspace{0.08cm} \Omega\left(\frac {\gamma^3}{ \ln(8/\gamma)}\right)- O(\tau\gamma)
\right\}.
\]
\end{lemma}

\begin{proof}  We may assume without loss of generality that
  $\pos(f)+\negg(f)=1$ (by definition, we have $\pos(f)\le 1-\lambda$ and 
  $\negg(f)\ge \lambda$)~and $\theta \geq 0$.
Using Lemma~\ref{lem:closest-mon-ltf}, it suffices to lower bound $\dist(f,g)$ where $g(x) = \sign( \sum_{i \in P} w_ix_i - \theta)$.
Similar to the proof of Lemma \ref{lem:closest-mon-ltf}, we 
  view $x\in \{-1,1\}^n$ as the concatenation of $x'\in \{-1,1\}^P$
  and $y\in \{-1,1\}^N$.

The proof has two cases depending on whether or not $\pos(f) \geq 2/3$:\vspace{0.05cm}
\medskip

\noindent {\bf Case 1:} $\pos(f) \geq 2/3.$  We begin by observing that
\begin{equation} 
\label{eq:aaa}
\mathop{\Pr}_{\by}\left[\sum_{j \in N} w_j\cdot \by_j < -  \sqrt{\negg(f)} \right] > \frac{1}{8} - \frac{\tau}{\sqrt{\negg(f)}}\ge 
\frac{1}{8} - \frac{\tau}{\sqrt{\lambda}}\ge \frac{1}{16},
\end{equation}
using the Berry--Ess\'een theorem as well as the fact that a Gaussian distribution has at least $1/4$ of its mass at least one standard deviation away from its mean.
Next we establish the following:

\begin{claim} \label{claim:pos1}
If $\pos(f) \geq 2/3$ then $\Pr_{\bx'}\big[ 0 \leq \sum_{i \in P} w_i\cdot
  \bx'_i - \theta < \sqrt{\negg(f)} \big] \geq \Omega({\sqrt{\lambda} }  \gamma^2) - 3\tau.$  \end{claim}
\begin{proof}
This holds because (using Berry--Ess\'een) we have
\begin{align*}
\mathop{\Pr}_{\bx'}\left[0 \leq \sum_{i \in P} w_i\cdot\bx'_i - \theta < \sqrt{\negg(f)} \right] &\geq \mathop{\Pr}_{\bz \sim \calN(0, \pos(f))}\left[ \bz\in \big[\theta, \theta + \sqrt{\negg(f)}\big]\right] - {\frac {2\tau}{\sqrt{\pos(f)}}} \\
                   &\geq \sqrt{\negg(f)} \left(\frac{1}{\sqrt{2\pi \cdot \pos(f)}}\cdot \exp\left({-\frac{(\theta+\sqrt{\negg(f)})^2}{2\hspace{0.03cm}\pos(f)}}\right) \right) - 
                 3\tau \\
                   &\ge \sqrt{\frac{\lambda}{{2\pi}}}\cdot
                   \exp\left(- \frac{(\theta+\sqrt{1/3})^2}{4/3} \right)-3\tau =\Omega\big(\sqrt{\lambda}\gamma^2\big)-3\tau.
\end{align*}
where we used $\pos(f)\in [2/3,1]$, $\negg(f)\le 1/3$, $\theta \leq \sqrt{2\ln (2/\gamma)}$ (which we prove immediately) and   
$$
\exp\left(- \frac{(\theta+\sqrt{1/3})^2}{4/3} \right)=\Omega( \gamma^2).
$$
This can be shown by using $\theta \leq \sqrt{2\ln (2/\gamma)}$ 
  and considering the two cases of $\theta=O(1)$ and 
  $\theta=\omega(1)$.
The case when $\theta=O(1)$ is trivial since $\gamma\le 1$, and when $\theta=\omega(1)$, we have $ 3(\theta+\sqrt{1/3})^2/4< \theta^2$.
  
Finally, the upper bound of $\sqrt{2\ln (2/\gamma)}$
   on $\theta$ follows directly from 
\[
{\frac \gamma 2} \leq 
\mathop{\Pr}_{\bx}\left[ \sum_{i\in [n]} w_j\cdot  \bx_j  > \theta \right] \leq e^{-\theta^2/2},\]
where we have used Hoeffding's inequality and the fact that $\pos(f)+\negg(f)=1.$
\end{proof}

Combining Claim \ref{claim:pos1} and (\ref{eq:aaa}), we get that if $\pos(f) \geq  2/3$ then
\begin{align*}
\mathop{\Pr}_{\bx}[f(\bx) \neq g(\bx)] &\geq \mathop{\Pr}_{\by}\left[\sum_{j \in N} w_j\cdot \by_j <  -\sqrt{\negg(f)}\right] \cdot \mathop{\Pr}_{\bx'}\left[0 \leq \sum_{i \in P} w_i\cdot \bx'_i - \theta < \sqrt{\negg(f)} \right]\label{eq:lem3.1convlast},                   \end{align*}
which is at least 
$\Omega(\sqrt{\lambda}\gamma^2) -  O(\tau)$
as desired.\vspace{0.05cm}

\medskip
\noindent {\bf Case 2:}   $\pos(f) < 2/3$. 
We can assume $g$ satisfies $|\E[g]| \leq 1 - {\gamma}/{2} $ since otherwise $\dist(f,g) \geq \gamma/4$ by virtue of the difference in their expectations.
Because $|\E[g]| \leq 1 - {\gamma}/{2} $, 
  at least a~$ {\gamma}/{4}$ fraction~of $x' \in \{-1, 1\}^P$ satisfy $g(x') = 1$, i.e., $\theta \leq \sum_{i \in P} w_i\cdot x_i'.$
{By Hoeffding's inequality, we have
\begin{equation}\label{hehe1}
\mathop{\Pr}_{\bx'}\left[\sum_{i\in P}w_i\cdot \bx_i'\ge \sqrt{2 \ln (8/\gamma)} \cdot \sqrt{\pos(f)} \right]\le \gamma/8.
\end{equation}
This means that at least a $ \gamma/ 8$ fraction of $x' \in \{-1,1\}^P$ satisfy (since $\theta\ge 0$)
$$\theta \leq \sum_{i \in P} w_i\cdot x'_i \leq \theta + \sqrt{2 \ln (8/\gamma)} \cdot \sqrt{\pos(f)}.$$  Next,  recall that the variance of $\sum_{j \in N} w_j\cdot \by_j$ is $\negg(f)$. So the Gaussian tail lower bound (\ref{gatail}) from Section \ref{sec:useful-structural} together with Berry--Ess\'een as well as $\negg(f) \geq 1/3 > \pos(f)/2$ gives
\begin{align}
\mathop{\Pr}_{\by}\left[\sum_{j \in N} w_j\cdot \by_j < -\sqrt{2\ln(8/\gamma) \cdot \pos(f)}\right] \nonumber&\geq
\mathop{\Pr}_{\by}\left[\sum_{j \in N} w_j\cdot\by_j < -2\sqrt{ \ln(8/\gamma) \cdot \negg(f)}\right] \\[0.2ex]
 &\geq
 \Omega\left(\frac{\gamma^2}{\sqrt{\ln (8/\gamma)}}\right)
  - {\frac \tau {\sqrt{\negg(f)}}}\nonumber \\[0.3ex]
 &\geq 
 \Omega\left(\frac{\gamma^2}{\ln(8/\gamma)}\right)
 - 2 \tau .
\label{hehe2}\
 \end{align}
As a result, $\Pr_{\bx}[f(\bx)\ne g(\bx)]$ is at least 
  the product of (\ref{hehe1}) and (\ref{hehe2}),
  concluding the proof.
}
\end{proof}
  
\subsection{Restrictions of monotonically non-decreasing variables}
\def\brho{\boldsymbol{\rho}}

Our goal in this subsection is to show that for any \red{LTF} $f\colon\{-1,1\}^n\rightarrow
  \{-1,1\}$, a random restriction that fixes variables of $f$ that are 
  monotonically non-decreasing 
  has, in expectation, the same distance~to monotonicity as the original function $f$.  We will use this in the proof of Lemma \ref{lem:regularize-and-balance}.

\begin{lemma}\label{lem:restriction}
Let $f\colon\{-1,1\}^n\rightarrow \{-1,1\}$ be an \red{LTF} 
 and let 
  $S\subseteq [n]$~be~a set of variables of $f$ that are monotonically non-decreasing. 
Then a random restriction $\brho$ that~fixes each variable in $S$ independently and uniformly to a random
element of $\{-1,1\}$ satisfies $$
\mathop{\E}_{\brho}\Big[\hspace{0.02cm}\dist(f_{\brho},\textsc{Mono})\Big] = \dist(f, \textsc{Mono}).
$$
\end{lemma}

\red{
\begin{proof}
We let $g \colon \{-1, 1\}^n \to \{-1, 1\}$ be the LTF from Lemma~\ref{lem:closest-mon-ltf}. 
\begin{align*}
\Ex_{\brho}\Big[\dist(f_{\brho},\textsc{Mono})\Big] = \Ex_{\brho}\left[ \dist(f_{\brho}, g_{\brho})\right] 
									   = \dist(f, g) 
									   = \dist(f, \textsc{Mono}).
\end{align*}
\end{proof}
}

\section{Algorithmic tools for LTFs}
\label{sec:new-algorithmic} 
Our algorithm uses a few simple subroutines that may be viewed as relatively low-level algorithmic tools for working with LTFs. We present and analyze those   tools in this section.

\subsection{Finding high-influence variables}
\label{sec:find-heavy-weight}

We start with a subroutine that finds high-influence variables
  of a function. 
\begin{figure}[t]\begin{framed}
\noindent Subroutine {\tt Find-Hi-Influence-Vars}$(f,\rho,\tau,\delta)$\vspace{-0.16cm}
\begin{flushleft}
\noindent {\bf Input:} Black-box oracle access to  $f\colon \{-1,1\}^n \to \{-1,1\}$, a restriction $\rho \in \{-1,1,\ast\}^{n}$, and  parameters $\tau,\delta > 0$.

\noindent {\bf Output:}  Set $H \subseteq \stars(\rho).$

\begin{enumerate}

\item \vskip -.02in If $|\stars(\rho)|$ is not a power of 2, augment $\stars(\rho)$ with ``dummy'' variables (that are irrelevant in $f_\rho$) to bring its size to the next power of 2.  Let $X'$ be this set of variables.\vspace{-0.04cm}

\item \vskip -.02in Initialize the collection ${\cal S}$ of sets of variables to be ${\cal S} = \{ X' \}.$\vspace{-0.04cm}

\item \vskip -.02in While ${\cal S}$ contains an element (i.e., a set of variables) which is of size $>1$:\vspace{-0.04cm}
      \begin{enumerate}
        \item  Remove an arbitrary element $X$ from ${\cal S}$ that is of maximum size (any one will do).\vspace{0.04cm}
        \item  Partition $X$ into equal-size subsets $A$ and $B$ (any partition will do).\vspace{0.04cm}
        \item  Let $\delta' = { {\tau^2}\delta/({8 \log n}})$.
        Call {\tt Estimate-Sum-of-Squares}$(f_\rho,A,\tau^2/10,\delta')$   and let\\ $\hat{a}$ be the value it returns, and call {\tt Estimate-Sum-of-Squares}$(f_\rho,B,\tau^2/10,\delta')$\\ and let $\smash{\hat{b}}$ be the value it returns.
        (If the total number of calls made to {\tt Estimate-Sum-of-Squares} ever exceeds ${ {8 \log n}/{\tau^2}}$, halt and output
        ``fail.'')\vspace{0.06cm}
         \item  If $\hat{a} > 3\tau^2/4$ then add $A$ to ${\cal S}$.  Similarly if $\hat{b}>3\tau^2/4$ then add $B$ to ${\cal S}$.
      \end{enumerate} 

\item \vskip -.02in Return the set $H$ that consists of all elements $i$ such that the set $\{i\}$ is an element of ${\cal S}.$ 
\end{enumerate}
\end{flushleft}\vskip -0.1in
\end{framed}\vspace{-0.06in}
\caption{Subroutine {\tt Find-Hi-Influence-Vars}.}
\end{figure}

\begin{lemma} \label{lem:find-hi-influence-vars}
Suppose that the subroutine {\tt Find-Hi-Influence-Vars}$(f,\rho,\tau,\delta)$ is called on a function $f\colon \{-1,1\}^n \to \{-1,1\}$, a restriction $\rho \in \{-1,1,\ast\}^n$, and parameters $\tau,\delta>0$. Then it runs
in $\tilde{O}( { {\log n \cdot \log(1/\delta)}/{\tau^{10}}})\cdot n$ time, makes at most $\tilde{O}({ {\log n \cdot \log(1/\delta)}/{\tau^{10}}})$ queries, 
and with probability at least $1-\delta$ it outputs a set $H \subseteq \stars(\rho)$ such that 
\begin{itemize}
\item If $|\widehat{f_\rho}(i)| \geq \tau$ then $i \in H$;\vspace{-0.1cm}
\item If $|\widehat{f_\rho}(i)| < \tau/2$ then $i \notin H$.
\end{itemize}

\end{lemma}

\begin{proof}
The query complexity and running time of the subroutine are immediate given the ``escape condition'' in Step~3(c) and the bounds of Lemma \ref{lem:estdeg1}.  Below we establish the claimed performance guarantee.
First notice that given the ``escape condition'', with probability
  at least $1-\delta$, every call to {\tt Estimate-Sum-of-Squares} returns 
  an estimate that is additively $\pm \tau^2/10$-accurate. 
We assume this is the case and show below that the subroutine
  returns a set $H$ with the claimed property.
  
Next, since the size of the initial set $X'$ is a power of two (say $2^r$), it is clear that every set $X$ that ever belongs to
${\cal S}$ will have size $2^\ell$ for some integer $\ell \leq r.$  It is also clear that the elements of ${\cal S}$ are always disjoint sets.

We may divide the execution of the Step~3 loop into $r$ stages $0,\dots,r-1$ where in the $t$-th stage each execution of Step~3(a) selects a set $X$ of size $2^{r-t}.$  Right before the start of the $t$-th stage every element of ${\cal S}$ is a set of size $2^{r-t}$, and after the $t$-th stage every element is of size $2^{r-t-1}.$ 

Consider the state of the collection ${\cal S}$ right before the beginning of stage $t$. Because all calls to {\tt Estimate-Sum-of-Squares} in stage $t-1$ returned estimates that are additively $\pm \tau^2/10$-accurate,  every $X \in {\cal S}$ right before the beginning
of stage $t$ has $\sum_{i \in X} \widehat{f}(i)^2 \geq 0.65\tau^2.$~By Parseval's identity we have that $\smash{\sum_{i \in \stars(\rho)}\widehat{f}(i)^2 \leq 1}$  so since the elements of ${\cal S}$ are disjoint sets, it follows that before the beginning of Stage $t$ the number of elements of ${\cal S}$ is at~most~$1/(0.65\tau^2) \leq 2/\tau^2,$~and hence there are at most $4/\tau^2$ calls to {\tt Estimate-Sum-of-Squares} made in stage $t$.  Since the number of stages $r$ is at most $1+\log n < 2 \log n$, 
there are at most ${ {8 \log n}/{\tau^2}}$ calls to {\tt Estimate-Sum-of-Squares} made in total (so the ``escape condition'' in Step~3(c) is not triggered).

Finally, given that all calls to {\tt Estimate-Sum-of-Squares} return estimates that are additively $\pm \tau^2/10$-accurate, it is clear that the set $H$ will have the claimed property.
\end{proof}

\subsection{Checking that high-influence variables have positive weight} 
\label{sec:check-weight-positive}

Given an LTF,
the next subroutine checks whether the weight of a variable 
  is positive.

\begin{framed}
\noindent Subroutine 
{\tt Check-Weight-Positive}$(f,\rho,i,\tau,\delta)$
\vspace{-0.16cm}
\begin{flushleft}\noindent {\bf Input:} Black-box oracle access to an LTF $f(x) = \sign( \sum_{i=1}^n w_i\cdot x_i - \theta)$, a restriction $\rho$, 
  an element $i \in \stars(\rho)$, and  parameters $\tau,\delta > 0$.

\noindent {\bf Output:}  Either ``negative,'' ``positive,'' or ``fail.''

\begin{enumerate}

\item \vskip -.02in Draw $O( {\log(1/\delta)}/{\tau})$ uniform random edges from the $2^{| \stars(\rho)| - 1}$
edges in direction\\ $i$ that are consistent with $\rho$.\vspace{-0.06cm}

\item \vskip -.02in If $f$ (equivalently, $f_\rho$) is bi-chromatic and monotone on any of these edges, return ``positive;''
if $f$ is bi-chromatic and anti-monotone on any of these edges, return ``negative;'' if $f$ is not bi-chromatic on any of these
edges, return ``fail.''
\end{enumerate}
\end{flushleft}\vskip -0.1in
\end{framed}

(Note that since $f$ is an LTF, it is a unate function and thus it is impossible for the set of edges drawn in Step~1 to contain both a monotone edge and an anti-monotone edge.)

\begin{lemma} \label{lem:check-weight-positive}
Suppose that {\tt Check-Weight-Positive} is called on an LTF  $f(x)= \sign(\sum_{i=1}^n w_i x_i - \theta),$  a restriction $\rho \in \{-1,1,\ast\}^n$, $i \in  \stars(\rho)$, and $\tau,\delta>0$ such that $\smash{|\hat{f_\rho}(i)| \geq \tau}$  (note that the~latter implies that $w_i \neq 0$). Then it runs
in $O( {\log (1/\delta)}/{\tau}) \cdot n$ time, makes $O( {\log(1/\delta)}/{\tau})$ queries, and
\begin{itemize}
\item If it does not output ``fail'', which happens with probability at most $\delta$, \vspace{-0.12cm}
\item It outputs ``positive'' if $w_i > 0$, and it outputs ``negative'' if $w_i < 0$.   
\end{itemize}
\end{lemma}
\begin{proof}  This follows from  $\Inf_i(f)= |\hat{f}(i)|$ for any LTF as well as the  fact that if $\smash{\hat{f}(i) \neq 0}$ in an LTF $f=\sign(\sum_{i=1}^n w_i x_i - \theta),$  then $\smash{\hat{f}(i)>0}$ iff $w_i>0$ and $\smash{\hat{f}(i)<0}$ iff $w_i<0.$
\end{proof}



\section{Detailed description of the algorithm} \label{sec:detailed}

We present our algorithm and its analysis in this section.\vspace{-0.1cm}
\subsection{The algorithm}
Our main testing algorithm, {\tt Mono-Test-LTF}, is presented in Figure \ref{fig:main}.  Its main components are~two procedures called {\tt Regularize-and-Balance} and {\tt Main-Procedure}, described and analyzed in Sections \ref{sec:regularize-and-balance} and \ref{sec:main-procedure}.
As will become clear later, {\tt Mono-Test-LTF} is one-sided, i.e., it always outputs 
  ``monotone'' when the input function $f$ is monotone (because it only outputs
  ``non-monotone'' when an anti-monotone edge is found, via {\tt Check-Weight-Positive} or {\tt Edge-Tester}).
Thus,~our analysis of correctness below focuses on the case when $f$ is an LTF that is 
  $\eps$-far from monotone, and shows that in this case {\tt Mono-Test-LTF} outputs
  ``non-monotone'' with probability at least 2/3.

\begin{figure}
\begin{framed}
\noindent Algorithm  {\tt Mono-Test-LTF}$(f,\eps)$
\vspace{-0.16cm}
\begin{flushleft}\noindent {\bf Input:} Black-box oracle access to an LTF $f\colon \{-1,1\}^n \to \{-1,1\}$ and 
  a parameter $\eps > 0$.

\noindent {\bf Output:}  Returns ``monotone'' or ``non-monotone.'' 
\begin{enumerate}

\item \vskip -.02in 
Call {\tt Regularize-and-Balance}$(f,\eps)$.  If it returns a restriction $\rho\in \{-1,1,\ast\}^{[n]}$ then
continue to Step~2; if it returns ``non-monotone,'' halt and output ``non-monotone;''\\ if it returns ``monotone,'' halt and output ``monotone.''

\item \vskip -.02in Call {\tt Main-Procedure}$(f,\rho,\eps)$.  If 
  it returns ``non-monotone,'' halt and output ``non-monotone;''
  if it returns ``monotone,''  halt and output ``monotone.''

\end{enumerate}\end{flushleft}\vskip -0.1in
\end{framed}\vspace{-0.12cm}
\caption{Main algorithm {\tt Mono-Test-LTF}.
If $f$ is monotone it outputs ``monotone'' with\\ probability 1; if $f$ is $\eps$-far from monotone, it outputs ``non-monotone'' with probability $\geq 2/3$.}\label{fig:main}\end{figure} 

\subsection{Key properties of procedure {\tt Regularize-and-Balance}} \label{sec:regularize-and-balance}

Let $f\colon\{-1,1\}^n\rightarrow \{-1,1\}$ be an LTF,
  given by $f(x)=\sign(w\cdot x-\theta)$.
Assume that $f$ is $\eps$-far~from monotone.
The goal of~the procedure  {\tt Regularize-and-Balance}$\hspace{0.06cm}(f,\epsilon)$
  is to return a restriction $\rho\in $ $\{-1,1,*\}^{[n]}$ such that
  $f_\rho$ is a $(\tau,\eps,\lambda)$-non-monotone LTF (with respect to $(w,\theta)$), where 
\begin{equation}\label{parameters}
\lambda=\frac{\eps^{2}}{36\ln (12/\eps)}\quad\text{and}\quad\tau=\frac{\lambda\hspace{0.03cm}\eps}{\log^2 n}.
\end{equation}

Here is some intuition that may be helpful in understanding   {\tt Regularize-and-Balance}.  If~the procedure halts and outputs ``monotone'' in
Step~2, this signals that the (low-probability) \mbox{failure} event~of {\tt Find-Hi-Influence-Variables} has taken place (since it has spuriously identified more variables as having high influence than is possible given Parseval's identity;
  see Lemma \ref{lem:find-hi-influence-vars}).
The procedure halts and outputs ``non-monotone'' in Step~3 only if {\tt Check-Weight-Positive} has unambiguously found an 
anti-monotone edge.  If the procedure outputs ``monotone'' in Step~3, this signals the (low-probability) event that {\tt Check-Weight-Positive} failed to identify some index $i \in H$ (which was supposed to have high influence) as either having $w_i>0$ or $w_i<0$.  Finally if it outputs ``monotone'' in Step~4, this signals that $f$ appears to be close to monotone.\hspace{0.05cm}\footnote{This will become clear later in the proof 
  of Lemma \ref{lem:regularize-and-balance} where we show that
  Step 4 fails with
  low probability when $f$ is far from monotone.}

\begin{figure}[t]\begin{framed}
\noindent Procedure  {\tt Regularize-and-Balance}$(f,
  \eps)$\vspace{-0.2cm}
\begin{flushleft}
\noindent {\bf Input:}
  Parameter $\eps>0$ and black-box oracle access to an LTF $f\colon\{-1,1\}^n \to\{-1,1\}$
  of the form $f(x)=\sign(w\cdot x-\theta)$, with unknown weights $w$ and threshold $\theta$.

\noindent {\bf Output:}  Either ``non-monotone,'' ``monotone,'' or a restriction $\rho \in \{-1,1,\ast\}^{{[n]}}$. 

\begin{enumerate}
\item \vskip -.02in Let $C_{RB}>0$ be a large enough constant, and let 
  $\tau'$ and $\delta$ be the following parameters:
$$ 
\tau'= \tau^2\eps^3/C_{RB}\quad\text{and}\quad \delta= \tau'^2/C_{RB}.
$$
\item \vskip -.02in Call {\tt Find-Hi-Influence-Vars}$(f,(*)^n, {\tau'},\delta)$ 
and let $H$ be the set it returns.\\
If $|H|>4/\tau'^2$, halt and output ``monotone.''

\item \vskip -.02in For each $i \in H$, call {\tt Check-Weight-Positive}$(f,(*)^n,i, {\tau'/2}
,\delta)$.  If any call returns ``negative,'' halt and output ``non-monotone;'' if any call returns ``fail,'' halt and output ``monotone;''  otherwise (the case that every call returns ``positive'') continue to Step~4.

\item \vskip -.02in  {Repeat $C_{RB}/\eps$ times: 
\begin{flushleft}\begin{enumerate}\item[]\vskip -.1in Draw a restriction $\rho$, which has support $H$ and 
  is obtained by selecting a random assignment from $\{-1,1\}^H$.
Call {\tt Check-Fourier-Regular}$(f_\rho,[n]\setminus H, \sqrt{\red{12}\tau'/\eps},\delta/2)$
and {\tt Estimate-Mean}$(f_\rho,\eps/6,\delta/2)$.
\end{enumerate}\end{flushleft}
\vskip -.09in Halt and output the first $\rho$ where {\tt Check-Fourier-Regular} outputs ``regular'' and {\tt Estimate-Mean}
  returns a number of absolute value $\le 1-7\eps/6$. If the procedure fails to find such a restriction $\rho$, halt and output ``monotone.''}
\end{enumerate}\end{flushleft}
\vskip -.1in
\end{framed}\vspace{-0.12cm}
\caption{Procedure {\tt Regularize-and-Balance}.  Our analysis (Lemma \ref{lem:regularize-and-balance}) focuses on the case\\ when $f$ is $\eps$-far from monotone.}\label{fig:regularizebalance}
\end{figure}

It is clear that {\tt Regularize-and-Balance} is one-sided.

\begin{fact}\label{one-sided-1}
{\tt Regularize-and-Balance}$(f,\eps)$ never returns ``non-monotone'' if $f$ is monotone.
\end{fact}

We also have the following upper bound for the number of queries it uses (which can be straight forwardly 
verified by tracing through procedure calls and parameter settings):
  
\begin{fact}\label{query-comp-1}
The number of queries used by {\tt Regularize-and-Balance}$(f,\eps)$ is $\tilde{O}\hspace{0.03cm}({\log^{41}n}/{\eps^{90}})$. 
\end{fact}

We now prove the main property of the procedure {\tt Regularize-and-Balance}.

\begin{lemma} \label{lem:regularize-and-balance}
If $f(x)=\sign(w\cdot x-\theta)$ is $\eps$-far from monotone,
  then with probability at least $9/10$,  {\tt Regularize-and-Balance}$(f,\eps)$ 
  returns either ``non-monotone,'' or a restriction $\rho$ such that 
  $f_\rho$ is~a $(\tau,\eps,\lambda)$-non-monotone LTF with respect to $(w,\theta)$.
\end{lemma}

\begin{proof}
Using Lemma~\ref{lem:find-hi-influence-vars}, with probability $1-\delta$, {\tt Find-Hi-Influence-Vars} in Step 2 returns
  a set $H\subseteq [n]$ of indices that satisfies the following property:
\begin{equation}\label{eventevent}
\text{If $|\hat{f}(i)| \geq \tau'$ then $i\in H$;  
If $|\hat{f}(i)| < \tau'/2$ then $i \notin H$.}
\end{equation}
When this happens,
 we have by Parseval $|H|\le 4/\tau'^2$, and the procedure continues to Step 3.

We consider two subevents:\hspace{-0.02cm}
  $E_0'$:\hspace{-0.02cm} $H$ satisfies (\ref{eventevent}) but  contains an elements $i$ with $w_i < 0$; and~$E_0$: $H$ satisfies (\ref{eventevent}) and every $i \in H$ has $w_i > 0$.
We have $\Pr[E_0']+\Pr[E_0]\ge 1-\delta$ as discussed above. Below we show that the procedure 
  returns ``non-monotone'' with high probability, conditioning on $E_0'$,
  and it returns a restriction with the desired property with high probability, conditioning on~$E_0$.
By the end we combine the two cases to conclude that
\begin{align*}
&\Pr[E_0']\cdot \Pr\big[\text{\hspace{0.02cm}the procedure returns ``non-monotone''}\hspace{0.05cm}|\hspace{0.05cm}E_0'\hspace{0.03cm}\big]\\[0.2ex]
&\ \ \ \ \ +\Pr[E_0]\cdot
\Pr\big[\text{\hspace{0.02cm}it returns $\rho$ such that $f_\rho$ is $(\tau,\eps,\lambda)$-non-monotone}\hspace{0.07cm}|\hspace{0.05cm}E_0\hspace{0.03cm}\big]\ge 9/10.
\end{align*}
 
We first address the (easier) case of $E_0'$.
Assume $i\in H$ satisfies $w_i<0$.
From (\ref{eventevent}), $|\hat{f}(i)| \ge \tau'/2$ and thus, {\tt Check-Weight-Positive}$(f, (*)^n, i, \tau'/2, \delta)$ in Step 3 returns ``negative'' with~probability $1-\delta$, 
and the procedure returns ``non-monotone'' with~probability 
   $1- \delta$, conditioning on $E_0'$.

Next we address the (harder) case of $E_0$. First we use $E_1$ to denote the event that every~call~to {\tt Check-Weight-Positive} in Step 3 returns the correct answer, i.e., it returns ``positive'' for every $i\in H$.
By a union bound we have $\Pr[E_1\hspace{0.03cm}|\hspace{0.03cm}E_0]\ge 
  1-4\delta/\tau'^2$. 
  
Assuming that $E_1$ happens, the procedure then proceeds to Step 4 and we use $E_2$ to denote the event 
  that {\tt Check-Fourier-Regular} and {\tt Estimate-Mean} 
  in Step 4 return the correct answer, i.e.:
\begin{enumerate}
\item {\tt Check-Fourier-Regular} outputs ``not regular''\vspace{0.005cm} if $|\hat{f}_{\rho}(i)|\ge \sqrt{\red{12}\tau'/\eps}$ for some $i\in [n]\setminus H$, and outputs ``regular'' if $\smash{|\hat{f}_\rho(i)|\le {\red{3}\tau'}/\eps}$ for all $i\in [n]\setminus H$, {for every $\rho$ in Step 4, and} \vspace{-0.06cm}
\item {\tt Estimate-Mean} returns a number $a$ with $|a-\E[f_\rho]|\le \eps/6$,
  for every $\rho$ in Step 4.  
\end{enumerate}
We also write $E_3$ to denote the event that one of the restrictions $\rho$ drawn
  in Step 4 satisfies~that~$f_\rho$ is both $(2\eps/3)$-far from monotone and $(3\tau'/ \eps )$-Fourier-regular.
By a union bound we have $$\Pr[E_2\hspace{0.03cm}|\hspace{0.03cm}E_0\land E_1] \ge 1-C_{RB}\delta/\eps.$$
In the rest of the proof we show that 1) $\Pr[E_3\hspace{0.03cm}|\hspace{0.03cm}E_0\land E_1]\ge 99/100$ and 
  2) given $E_0,E_1,E_2$~and~$E_3$ the procedure always returns a restriction
  $\rho$ such that $f_\rho$ is $(\tau,\eps,\lambda)$-non-monotone.
Together we~have that the procedure returns such a $\rho$ with probability
  at least (conditioning on $E_0$) 
  $$
(1-4\delta/\tau'^2) \cdot (1-C_{RB}\delta/\eps-1/100).
$$

 Summarizing the two cases of $E_0'$ and $E_0$  we have that {\tt Regularize-and-Balance}
  returns either ``non-monotone'' or
  a restriction $\rho$ such that $f_\rho$ is $(\tau,\eps,\lambda)$-non-monotone with 
  probability at least
$$
\Pr[E_0']\cdot (1-\delta)+\Pr[E_0]  
  \cdot (1-4\delta/\tau'^2) \cdot (1-C_{RB}\delta/\eps-1/100) > 9/10,
$$
using $\Pr[E_0']+\Pr[E_0]\ge 1-\delta$ and our choice of $\delta$ (by letting $C_{RB}$ be large
  enough).

We use the following claim to show that $\Pr[E_3\hspace{0.03cm}|\hspace{0.03cm}E_0\land E_1]\ge 99/100$.

\begin{claim}\label{hehe}
A random restriction $\rho$ over $H$ satisfies that $f_\rho$ is both $(2\eps/3)$-far from monotone and 
  $( 3\tau'/ \eps)$-Fourier-regular with probability at least $\eps/3$.
\end{claim}
\begin{proof}
For each of the two properties, we have
\begin{flushleft}\begin{enumerate}
\item Proposition \ref{prop:its-cool}:
 with probability at least $1-(\eps/3)$, $f_\rho$ is $( 3\tau'/ \eps)$-Fourier-regular.\vspace{-0.08cm}
\item Lemma \ref{lem:restriction}: with probability at least $2\eps/3$, $f_\rho$ is $(2\eps/3)$-far
  from monotone. To see this, let $c$  denote the probability of 
  $f_\rho$ being $(2\eps/3)$-far from monotone. Then $c\ge 2\eps/3$ follows from
$$
(1-c)\cdot (2\eps/3)+c\cdot (1/2) \ge \eps,
$$
where we used the fact that distance to monotonicity is always at most $ {1}/{2}$.
\end{enumerate}\end{flushleft}
The claim then follows from a union bound.
\end{proof}

By choosing $C_{RB}$ to be a large enough constant,
  we have $\Pr[E_3\hspace{0.03cm}|\hspace{0.03cm}E_0\land E_1]\ge 99/100$.

Finally we show that conditioning on all four events $E_0,E_1,E_2,E_3$  
  the procedure always returns a restriction $\rho$ such that $f_{\rho}$ is a $(\tau, \eps, \lambda)$-non-monotone LTF. 
We do this in two steps: 
\begin{flushleft}\begin{enumerate}
\item First, given $E_3$, one of the restrictions $\rho$ drawn in Step 4
  is both $(2\eps/3)$-far from monotone and $\smash{(3\tau'/\eps)}$-Fourier-regular.
Given $E_2$,
  $\rho$ must 
  pass both tests, i.e., {\tt Check-Fourier-Regular} outputs ``regular'' and {\tt Estimate-Mean} returns a number of absolute value at most $\smash{1 - 7\eps / 6}$ in Step 4. The former is trivial; to see the latter, note that being 
  $(2\eps/3)$-far from monotone implies that $|\E[f_\rho]|\le 1-4\eps/3$ and therefore, the number
  returned by {\tt Estimate-Mean} is at most $1-7 \eps/6$, given $E_2$.
\ignore{Moreover, $f_\rho$ must be a $(\tau,\eps,\lambda)$-non-monotone LTF.
This follows from Theorem \ref{Dzindzalieta} and Lemma \ref{lem:reg-neg-coeff} (using our choice of $\lambda$).
(One can think of this as a completeness property of the procedure which guarantees the algorithm ``makes progress.")}
\item Second,  we show that if a restriction $\rho$ passes both tests in Step 4 of the procedure, then $f_{\rho}$ must be $(\tau, \eps, \lambda)$-non-monotone. One can think of this as a soundness property, saying that if the procedure halts and returns some restriction $\rho$, that it returns a correct one.
To see this, note that by $E_2$, $f_{\rho}$ is both $\smash{\sqrt{\red{12}\tau'/\eps}}$-Fourier regular and $\eps$-balanced. By Theorem~\ref{Dzindzalieta}, $f_\rho$ is $O(\sqrt{ \tau'/\eps^3})$-weight-regular, and
  $\tau$-weight-regular by letting $C_{RB}$ be large enough. 
It also follows from Lemma \ref{lem:reg-neg-coeff} that $f_\rho$ has $\lambda$-significant squared negative weights.
\end{enumerate}\end{flushleft} 
This finishes the proof of the lemma.
\end{proof}

\subsection{Key properties of {\tt Main-Procedure}} \label{sec:main-procedure}

{\tt Main-Procedure} is presented in Figure \ref{fig:mainprocedure}.
Given Lemma \ref{lem:regularize-and-balance}  we may assume that the input 
   $(f,\rho ,\eps)$ satisfies
  that $f_\rho$ is a $(\tau,\eps,\lambda)$-non-monotone LTF (see the choices of $\tau$ and $\lambda$ 
  in (\ref{parameters})).
  
We prove the following main lemma in this section.

\begin{figure}[t]\begin{framed}
\noindent Procedure  {\tt Main-Procedure}$(f,\rho,\eps)$\vspace{-0.16cm} 
\begin{flushleft}
\noindent {\bf Input:} 
Parameter $\eps>0$, black-box oracle access to an LTF $f \colon \{-1, 1\}^n\rightarrow \{-1, 1\}$ 
  of the form $f(x) = \sign(w\cdot x-\theta)$ with unknown weights $w$ and threshold $\theta$,
  and a restriction $\rho$.\\
\noindent {\bf Output:}  Either  ``non-monotone'' {or ``monotone.'' }
\begin{enumerate}
\item \vskip -.02in Set $t=0$ and $\rho^{(0)}=\rho$.
\item \vskip -.02in While {$|\stars(\rho^{(t)})| \geq 1/\tau^2$}, repeat the following steps:

\begin{enumerate}

\item\vskip -.02in Construct a subset $A_t \subseteq \stars(\rho^{(t)})$ by independently putting each index $i \in
\stars(\rho^{(t)})$ into $A_t$ with probability $1/{2}$.
\vspace{0.1cm}

\item Call {\tt Find-Balanced-Restriction}$(f,\rho^{(t)},A_t,\eps)$.  If it returns ``monotone''\\ then halt and return ``monotone;'' 
otherwise, it returns a restriction $\rho'$ with $\supp(\rho')=\supp(\rho^{(t)}) \cup A_t$.
\vspace{0.1cm}

\item Call {\tt Maintain-Regular-and-Balance}$(f,\rho',\eps)$.  If it returns
``non-monotone''\\ then halt and output ``non-monotone;'' if it returns
``monotone'' then halt and output ``monotone;''  otherwise, it returns a restriction $\eta$ and
  we set $\rho^{(t+1)}$ to ${\rho' \eta}$.\vspace{0.1cm}

\item Increment $t$ by 1.
{If $t > 4\log n$, halt and output ``monotone;''} otherwise proceed\\ to the next iteration of step (a) of the loop,  
\end{enumerate}

\item \vskip -.02in 
Let {$\smash{\eps'= { \eps^{3}}/(C{\log(1/\eps)}})$} for some large constant $C$; run {\tt Edge-Tester}$\smash{(f_{\rho^{(t)}},\eps',1/10)}$ \\and output what it outputs (either ``monotone'' or ``non-monotone''). 
\end{enumerate}\end{flushleft}

\vskip -.1in
\end{framed}\vspace{-0.12cm}\caption{Procedure {\tt Main-Procedure}.   Our analysis in Section \ref{sec:main-procedure} 
  focuses on the case when $f_\rho$ is\\ a $(\tau,\eps,\lambda)$-non-monotone LTF.}\label{fig:mainprocedure}
\end{figure}

\begin{figure}\begin{framed}
\noindent Subroutine  {\tt Find-Balanced-Restriction}$(f,\rho^{(t)},A_t,\eps)$\vspace{-0.16cm} 
\begin{flushleft}
\noindent {\bf Input:} oracle access to $f \colon \{-1, 1\}^n\rightarrow \{-1, 1\}$, restriction
$\rho^{(t)}$, $A_t \subseteq \stars(\rho^{(t)}),$ and $\eps > 0.$

\noindent {\bf Output:} ``monotone'' or a restriction $\rho'$ with $\supp(\rho')=\supp(\rho^{(t)})\cup A_t$
  that extends $\rho^{(t)}$.
\begin{enumerate}
\item[] {Repeat $C_{BR}\cdot \log n/\eps^3$ times for some large enough constant $C_{BR}$: 
\begin{flushleft}\begin{enumerate}\item[]\vskip -.1in Draw a restriction $\rho^*$, which has support $A_t$ and 
  is obtained by selecting a random assignment from $\{-1,1\}^{A_t}$, and let $\rho'=\rho^{(t)} \rho^*$.
Call 
{\tt Estimate-Mean}$(f_{\rho'},0.01,\delta)$, where $\delta=\eps^3/(200\hspace{0.03cm}C_{BR}\log^2n)$.
If it returns a number of absolute value at most $0.03$, halt and output $\rho'$.
\end{enumerate}\end{flushleft}
\vskip -.08in Otherwise, output ``monotone.''}
\end{enumerate}
\end{flushleft}\vskip -.1in
\end{framed}\vspace{-0.12cm}
\caption{Subroutine {\tt Find-Balanced-Restriction}. We are interested in the case when $f_{\rho^{(t)}}$ is a\\ $(\tau,\eps ,\lambda(1-t/(8\log n)))$-non-monotone LTF, 
  and $A_t$ satisfies the conditions of Lemma~\ref{lem:FBR}.\ignore{\red{Li-Yang: This used to say ``and $A_t$ and $B_t$ satisfy both $|A_t|\ge m/4$ and  (\ref{eq:1}).'' but the definitions of $B_t$ and $m$ do not appear till the next page.}}}\label{fig:findbalancedrestriction}
\end{figure}

\begin{lemma} \label{lem:main-procedure}
{\tt Main-Procedure}$(f,\rho,\eps)$ never returns ``non-monotone'' when
  $f$ is monotone.
When $f_\rho$ is a $(\tau,\eps ,\lambda)$-non-monotone LTF, it
  returns ``non-monotone'' with probability at least $81/100.$
\end{lemma}

The procedure only returns ``non-monotone'' when~it finds 
  an anti-monotone edge in the subroutine {\tt Check-Weight-Positive}. Hence we may focus on the case when $f_{\rho}$ is $(\tau, \eps, \lambda)$-non-monotone.
For this purpose, we analyze the three steps $\text{2(a)}, \text{2(b)}, \text{2(c)}$ of each while loop of {\tt Main-Procedure}, and prove
  the following lemma.
 
\begin{lemma} \label{key-lemma}
Let $t\le 4\log n$. Suppose that at the beginning of the $(t+1)$th loop of {\tt Main-Procedure}, 
  $f_{\rho^{(t)}}$~is $\smash{(\tau,\eps,\lambda(1-t/(8\log n)))}$-non-monotone.
Then with probability at least $\smash{1-1/(40\log n)}$,~it~either  returns ``non-monotone'' within this loop 
  or obtains~a set $\smash{A_t\subseteq [n]\setminus \supp(\rho^{(t)})}$ and a restriction $\smash{\rho^{(t+1)}}$  extending $\smash{\rho^{(t)}}$
  at the end of this loop such that 
\begin{enumerate}
\item $|A_t|\ge |\stars(\rho^{(t)})|/4$; \vspace{-0.08cm}
\item $\supp(\rho^{(t)})\cup A_t\subseteq \supp(\rho^{(t+1)})$;\ignore{\enote{Is it clear here that not only does the next restriction have support of the previous restriction and $A_t$, but also maintains the values at $\rho^{(t)}$? i.e the restriction $\rho^{(t+1)}$ ``extends" $\rho^{(t)}$? This should be clear from the proof, but maybe it will make it more readable to include it in the lemma statement or algorithm? Might make it easier to see what's going on.}} and\vspace{-0.1cm}
\item $f_{\rho^{(t+1)}}$
  is a $(\tau,\eps,\lambda(1-(t+1)/(8\log n)))$-non-monotone LTF.  
\end{enumerate}
\end{lemma} 

We use Lemma \ref{key-lemma} to prove Lemma \ref{lem:main-procedure}.

\begin{proof}[Proof of the Second Part of Lemma \ref{lem:main-procedure} using Lemma \ref{key-lemma}]
We consider the event $E$ where the conclusion of Lemma \ref{key-lemma} holds for every iteration of the while loop of {\tt Main-Procedure}.
As the condition of Lemma \ref{key-lemma} holds for the first loop ($t=0$) and there 
  are at most $4\log n$ many loops, this happens with probability at least $9/10$. {Since $E$ implies  $|A_t| \geq |\stars(\rho^{(t)})|/4$, we can also assume that the procedure never halts and outputs ``monotone'' due to line $\text{2(d)}$.}

Given $E$, {\tt Main-Procedure} either returns ``non-monotone'' as desired 
  or reaches line 3. Further\-more, if it reaches line 3, $\smash{f_{\rho^{(t)}}}$ must be $(\tau,\eps,\lambda/2)$-non-monotone by Lemma \ref{key-lemma} and have at~most $ {1}/{\tau^2}$ variables. It  follows from
  Lemma~\ref{lem:lem3.1-conv} that $f_{\rho^{(t)}}$  is $\eps'$-far from monotone, where~{$\eps'=\eps^{3}/(C{\log(1/\eps)})$}
  for some large enough constant~$C$. Finally, by Theorem~\ref{fact:edge-tester}, {\tt Edge-Tester} outputs ``non-monotone'' (by finding an anti-monotone edge) with probability at least $9/10$ and the proof is complete. 
  \ignore{
  \xnote{We also need to say that with high probability,
  the restriction $\rho^{(t)}$ in Step 3 has less than $\log n$ stars; Lemma 5.6 
  actually implies that most likely we do not reach Step 3. 
Should we change the number of rounds in Step 2 to something like 
  $\log n-\log \log n$? This seems a little awkward.}}
\end{proof}

\subsubsection{Proof of Lemma \ref{key-lemma}}

The proof of Lemma \ref{key-lemma} consists of three lemmas, one for each steps
  2(a), 2(b) and 2(c). 
Below~we assume that the condition of Lemma \ref{key-lemma} holds at the beginning
  of the $(t+1)^{\text{th}}$ loop, for some~$t\le 4\log n$.
We introduce the following notation for convenience.
We 
  let $I=\stars(\rho^{(t)})$, with $m=|I|$.
Given the random subset $A_t$ of $I$ found in Step 2(a), we let $B_t=I\setminus A_t$.
Also note that $\smash{m\ge 1/\tau^2}$. 

We start with the lemma for Step 2(a), which states that with high probability,
  $A_t$ is large and splits the weights (both positive and negative) in $I$ evenly.

\begin{lemma} \label{lem:a}
Assume that $f_{\rho^{(t)}}$ is a $(\tau,\eps,\lambda(1-t/(8\log n))$-non-monotone LTF. With probability~at least $\smash{1-\exp(-\Omega(\log^2 n))}$, $A_t$ and $B_t$ satisfy $|A_t|\ge m/{4}$,
\begin{equation}\label{eq:1}
\frac{1}{2}-\frac{1}{32\log n}\le \frac{\sum_{i\in A_t} w_i^2}{\sum_{i\in I} w_i^2}\le \frac{1}{2}+\frac{1}{32\log n}
\ \quad\text{and}\ \quad 
\frac{\sum_{i\in B_t: w_i<0} w_i^2}{\sum_{i\in B_t}\hspace{-0.03cm} w_i^2}\ge \lambda\left(1-\frac{t+1}{8\log n}\right).
\end{equation}
\end{lemma}
\begin{proof}
We consider the three events separately and then apply a union bound.

First by Chernoff bound, $|A_t|\ge m/4$ holds with probability at least $1-e^{-\Omega(m)}$.

Next for the first inequality in (\ref{eq:1}), assume without loss of generality that 
  $\sum_{i\in I} w_i^2=1$ {(as $f_{\rho^{(t)}}$ cannot
  be all-$1$ or all-$(-1)$)}. By Hoeffding bound the probability
  that it does not hold is at most 
$$
2 \exp\left(-\Omega\left(\frac{1/ \log^2 n }{\sum_{i\in I} w_i^4}\right)\right).$$
Since $f_{\rho^{(t)}}$ is $\tau$-weight-regular (over $I$), we have that $|w_i|\le \tau$ for all $i\in I$ and thus,
$$
\sum_{i\in I} w_i^4 \le \tau^2\cdot \sum_{i\in I} w_i^2=\tau^2.
$$
As a result, the second inequality holds with probability at least $1-\exp({-\Omega(1/(\tau^2\log^2 n))})$.

For the last inequality, note that 
  $\sum_{i\in I:w_i<0} w_i^{2}\ge \lambda(1-t/(8\log n))$.
Similarly by Hoeffding, 
$$
\Pr\left[\sum_{i\in B_t:w_i<0}w_i^2 <\left(\frac{\lambda}{2}\right)\left(1-\frac{t+0.5}{8\log n}\right)\right]
\le \exp\left(-\Omega\left(\frac{\lambda^2/ \log^2n }{\sum_{i\in I:w_i<0}w_i^4}\right)\right)\le \exp\left(-\Omega\left( \log^2n\right)\right).
$$
Combining the above with the analysis of the first inequality in (\ref{eq:1}), 
  the last inequality holds with probability at least $1-\exp(-\Omega(\log^2 n))$.
The lemma follows from a union bound.
\end{proof}

We give {\tt Find-Balanced-Restriction} in Figure \ref{fig:findbalancedrestriction} and show the following lemma for Step 2(b).  {(The {\tt Find-Balanced-Restriction} subroutine is similar to Algorithm~1 
  of \cite{RonServedio15}, and Lemma \ref{lem:FBR} and its proof are reminiscent of Lemma~7 of \cite{RonServedio15}; 
  however, because of some technical differences we cannot directly apply those results, so we give a self-contained presentation here.)} 

\begin{lemma} \label{lem:FBR}
Assume that $f_{\rho^{(t)}}$ is a $(\tau,\eps,\lambda(1-t/(8\log n)))$-non-monotone LTF,
  and sets $A_t$ and $B_t$ satisfy $|A_t|\ge m/4$ and \emph{(\ref{eq:1})}.
With probability at least $1/(100\log n)$, {\tt Find-Balanced-Restriction} 
  outputs a $\rho'$ with
  $\supp(\rho')=\supp(\rho^{(t)})\cup A_t$ such that $\rho'$ extends $\rho^{(t)}$ and
  $f_{\rho'}$ is $0.96$-balanced.
\end{lemma}

\begin{proof}
For convenience we use $f'$ to denote $f_{\rho^{(t)}}$, $w'$ to denote the
  weight vector $w$ but restricted~on $I$, and $\theta'$ to denote the new threshold, i.e.,
$$
\theta'=\theta-\sum_{i\in \supp(\rho^{(t)})} \rho^{(t)}(i)\cdot w_i.
$$
Without loss of generality we assume that $\sum_{i\in I} w_i'^2=1$. \red{We may additionally assume that $\theta'\ge 0$. This assumption is without loss of generality, because 1) if $\rho'$ is a 0.96-balanced restriction when $-\theta' \ge 0$, then $-\rho'$ is a 0.96-balanced restriction for $\theta' \le 0$, and 2) {\tt Find-Balanced-Restriction} will test the only take into account the absolute value of the output of {\tt Estimate-Mean}.}
Let $$\alpha=\sum_{i\in A_t} w_i'^2\quad\ \text{and}\ \quad\beta=\sum_{i\in B_t} w_i'^2.$$
We also use $a=b\pm c$ to denote the inequalities $b-c\le a\le b+c$.
Then from (\ref{eq:1}) we have~that $\alpha,\beta=1/2\pm O(1/\log n)$.
By the assumption of the lemma, $f'$ is 
  $\tau$-weight-regular and $\eps$-balanced.

For the analysis we define two events $E_1$ and $E_2$.
Here $E_1$ denotes the event that every call to {\tt Estimate-Mean}
  returns a number $a$ such that $|a-\E[f_{\rho'}]|\le 0.01$.
By a union bound, this happens with probability $1-1/(200\log n)$.
Let $E_2$ be the event that one of the restrictions $\rho^*$ drawn
  has $f'_{\rho^*}$ being $0.98$-balanced.
When $E_1$ and $E_2$ both occur, the subroutine 
  outputs a restriction $\rho'$ such~that $f_{\rho'}$ is $0.96$-balanced.
In the rest of the proof we show that $E_2$ happens with high probability.

To analyze the probability of $f'_{\rho^*}$  
  being $0.98$-balanced,
we use $\bx_i$ to denote an independent and unbiased random $\{{-1},1\}$-variable 
  for each $i\in I$, and let
$$
\bx_A=\sum_{i\in A_t} \bx_i\cdot w_i',\quad\bx_B=\sum_{i\in B_t} \bx_i\cdot w_i'\quad\text{and}\quad \bx=\bx_A+\bx_B.
$$
By Hoeffding bound and the assumption that $f'$ is $\eps$-balanced, we have
\begin{equation}\label{erik}
\red{2\eps}=\Pr[\bx\ge \theta']\le \exp(-\theta'^2/2).
\end{equation}

Using Berry--Ess\'een 
$\bx_A+\bx_B$ is $O(\tau)$-close to 
  a standard $\calN(0,1)$ Gaussian random variable,~denoted by  $\calG$, $\bx_A$ is $O(\tau)$-close to
  $\sqrt{\alpha}\hspace{0.03cm}\calG$, and $\bx_B$ is $O(\tau)$-close to $\sqrt{\beta}\hspace{0.03cm}\calG$.
  
Let $\theta^*>0$ be the threshold such that $\smash{\Pr[\hspace{0.03cm}|\sqrt{\beta}
  \hspace{0.03cm}\calG|\le \theta^*\hspace{0.03cm}]=0.01}$. Then 
$$
\Pr\big[\hspace{0.015cm}f'_{\rho^*}\ \text{is $0.98$-balanced}\hspace{0.03cm}\big]\ge \Pr\big[\bx_A\in [\theta'-\theta^*,\theta'+\theta^*]\big].
$$
This is because, for any number $x_A\in [\theta'-\theta^*,\theta'+\theta^*]$, we have
$$
0.495-O(\tau)\le \Pr\big[\bx_B\ge \theta'-x_A\big]=\Pr\left[\sqrt{\beta}\hspace{0.03cm}\calG\ge \theta'-x_A\right]\pm O(\tau)\le 0.505
  +O(\tau),
$$
in which case the function $f'_{\rho^*}$ is 
  $0.99-O(\tau)=0.98$-balanced.
To bound $\Pr\left[\hspace{0.03cm}\bx_A\in [\theta'-\theta^*,\theta'+\theta^*]\hspace{0.03cm}\right]$,
  we note that $\theta'\ge 0$ (by assumption) and $\theta^*=\Omega(1)$ (by our choice of $\theta^*$ and $\beta> 1/3$).
As a result,
$$
\Pr\big[\bx_A\in [\theta'-\theta^*,\theta']\big]\ge \Pr\big[
\sqrt{\alpha}\hspace{0.03cm}\calG\in [\theta'-\theta^*,\theta']\big]-O(\tau)=\Omega(1)\cdot
  \Omega(\eps^3)-O(\tau)=\Omega(\eps^3),
$$
where we used $\alpha > {1}/{3}$ by (\ref{eq:1}), $\tau=o(\eps^3)$, and $\exp(-\theta'^2/2)=\Omega(\eps)$ from (\ref{erik}) to obtain
$$
\min\Big(\exp\left(-(\theta^*)^2/(2\alpha)\right),\exp\left(-\theta'^2/(2\alpha)\right)\Big)=\Omega(\eps^3).
$$

As a result,  a random restriction $\rho^*$ is $0.98$-balanced with probability at least $\Omega(\eps^3)$.
Thus with probability $1-1/n$ (by choosing a large enough constant $C_{BR}$),  {\tt Find-Balanced-Restriction}
  gets such a restriction that would  pass the {\tt Estimate-Mean} test.
By a union bound on the two events~$E_1$ and $E_2$, {\tt Find-Balanced-Restriction}~returns a $0.96$-balanced
  $\rho'$ with probability at least
$$
1-1/(200\log n)-1/n> 1-1/(100\log n).
$$
This finishes the proof of the lemma.
\end{proof}

\begin{figure}[t]
\begin{framed}
\noindent Subroutine {\tt Maintain-Regular-and-Balanced}$(f,
  \rho',\eps )$\vspace{-0.2cm} 
\begin{flushleft}
\noindent {\bf Input:} oracle access to function $f\colon \{-1,1\}^n \to \{-1,1\}$, restriction $\rho'$, parameter $\eps>0.$

\noindent {\bf Output:} ``non-monotone,'' ``monotone,'' or a restriction $\eta$
  with $\supp(\eta)\subseteq B_t$ {extending $\rho'$}. 
  
  \begin{enumerate}
\item \vskip -.02in Let $C_M>0$ be a large enough constant; let 
  $\tau',\delta$ and $\tau^*$ be the following parameters:
$$ 
 \tau'= (\tau \eps /C_M)^2\cdot  {\sqrt{\lambda}},\quad {\delta = \tau'^2/(C_M\log n)\quad\text{and}\quad  {\tau^*=\tau'/\sqrt{\lambda}}.}
$$
\item \vskip -.02in Call {\tt Find-Hi-Influence-Vars}$(f,\rho', \tau',\delta )$ 
and let $H$ be the set that it returns. \\
If $\smash{|H|> 4/\tau'^2}$, halt and return ``monotone.''

\item \vskip -.02in For each $i \in H$, call {\tt Check-Weight-Positive}$(f,\rho',i,\tau'/2,\delta )$.  If any  call returns ``negative'' then halt and output ``non-monotone;'' if any call returns ``fail'' then halt and output ``monotone;''  otherwise (every call returns ``positive'') continue to Step~4.

\item \vskip -.02in  {Repeat {$C_M\log n/\sqrt{\lambda}$} times: 
\begin{flushleft}
\begin{enumerate}\item[]\vskip -.1in Draw a restriction $\eta$ with support $H$, by selecting a random assignment from $\{-1,1\}^H$.
Call {\tt Check-Fourier-Regular}$\smash{(f_{\rho'\eta},[n]\setminus \supp(\rho'\eta),\sqrt{C_M\tau^*},\delta/2 )}$
and {\tt Estimate-Mean}$(f_{ {\rho' \eta}},\eps/6,\delta/2 )$.
\end{enumerate}
\end{flushleft}
\vskip -.08in Halt and output the first restriction $\eta$ where {\tt Check-Fourier-Regular} outputs ``regular'' and {\tt Estimate-Mean}
  returns a number of absolute value $\le 1-7\eps/6$. If the procedure\\ fails to find such a restriction $\eta$, halt and output ``monotone.''}

\end{enumerate}\end{flushleft}
\vskip -.1in
\end{framed}\vspace{-0.12cm}
\caption{Subroutine {\tt Maintain-Regular-and-Balanced}.  Lemma \ref{lem:lastlem} assumes that $f_{\rho^{(t)}}$ is an\\  
  $\smash{(\tau,\eps,\lambda(1-t/(8\log n)))}$-non-monotone LTF,   
  $|A_t|\ge m/4$ and (\ref{eq:1}),
  and $f_{\rho'}$ is $0.96$-balanced.}\label{fig:maintain}
\end{figure}

For Step~2(c) of {\tt Main-Procedure}, the subroutine {\tt Maintain-Regular-and-Balanced} is given in Figure \ref{fig:maintain}.
It is very similar to {\tt Regularize-and-Balance} except the number of rounds in Step~4 and the choice of parameters $\tau'$ and $\delta$.
We show the following lemma.

\begin{lemma}\label{lem:lastlem}
Assume that $f_{\rho^{(t)}}$ is a $(\tau,\eps ,\lambda(1-t/(8\log n)))$-non-monotone LTF,
  $A_t$ and~$B_t$~\mbox{satisfy} both $|A_t|\ge m/4$ and \emph{(\ref{eq:1})}, and
  $f_{\rho'}$ is $0.96$-balanced.
Then with probability at least $1-1/(100\log n)$, 
 {\tt Maintain-Regular-and-Balance} 
  returns either ``non-monotone,'' or a restriction $\eta$ with $\smash{\supp(\eta)}$ $\smash{\subseteq B_t}$ such that
  $\smash{f_{\rho^{(t+1)}}}$, where $\smash{\rho^{(t+1)}=\rho' \eta}$, is 
  $\smash{(\tau,\eps,\lambda(1-(t+1)/(8\log n)))}$-non-monotone.
\end{lemma}

\begin{proof}
Most of the proof is similar to that of Lemma \ref{lem:regularize-and-balance}.
Let $\lambda'=\lambda(1-(t+1)/(8\log n)))$.

First 
  it follows from 
  Lemma \ref{lem:find-hi-influence-vars} that with probability at least $1-\delta$,   
  $H$ satisfies
\begin{equation}\label{eventevent2}
\text{If $|\hat{f}_{{\rho'}}(i)|\ge \tau'$, then $i\in H$; If $|\hat{f}_{{\rho'}}(i)|<\tau'/2$, then $i\notin H$.}
\end{equation}
When this happens, we have by Parseval $|H|\le 4/\tau'^2$,
  and the subroutine proceeds to Step 3.
   
Similar to the proof of Lemma \ref{lem:regularize-and-balance}, we
  consider two subevents: $E_0'$: $H$ satisfies (\ref{eventevent2}) but has an element
  $i$ with $w_i<0$; and $E_0$: $H$ satisfies (\ref{eventevent2}) and all $i\in H$ have $w_i>0$.
Then $\Pr[E_0']+\Pr[E_0]\ge 1-\delta$.
It is easy to show that, given $E_0'$, the~subroutine returns ``non-monotone''
  with probability~at least $1-\delta$.
So in the rest of the proof, we focus on the event of $E_0$ and  
  lowerbound the probability that the subroutine returns an
  $\eta$ such that $f_{\rho^{(t+1)}}$ is $(\tau, \eps, \lambda')$-non-monotone, conditioning on $E_0$. 
  
First we let $E_1$ be the event that {\tt Check-Weight-Positive} returns the correct answer in 
  Step~3, i.e., it returns ``positive''
for all $i\in H$. By a union bound, we have $\Pr[E_1\hspace{0.03cm}|\hspace{0.03cm}E_0]\ge   
  1-O(\delta/\tau'^2)$.
  
Conditioning on $E_0$ and $E_1$ 
  we proceed to Step 4 and introduce two events $E_2$ and $E_3$. 
Here~$E_2$ is the event that every call to {\tt Check-Fourier-Regular} and {\tt Estimate-Mean}
  returns the correct~answer (see the proof of Lemma \ref{lem:regularize-and-balance}).
By a union bound, 
$\Pr[E_2\hspace{0.03cm}|\hspace{0.03cm}E_0\land E_1]\ge 1-O(\tau'^2/\sqrt{\lambda})$.
To define $E_3$, we state the following claim to introduce the constant $a>0$, but
  delay its proof to the end.
  
\begin{claim}\label{hard}
\hspace{-0.02cm} Assume $E_0$ holds.\hspace{-0.02cm}
Then there is an absolute constant $a>0$ such that when $\eta$ is~drawn 
  uniformly at random from $\{-1,1\}^H$, $f_{\rho' \eta}$ is $(4\eps/3)$-balanced
  with probability at least $a\sqrt{\lambda}$.
\end{claim}  

Next, we~let~$E_3$ be the event that one of the restrictions $\eta$ drawn in
  Step 4 has $f_{\rho'\eta}$ being both $(4\eps/3)$-balanced and 
  $\smash{2\tau^*/a=2\tau'/(a\sqrt{\lambda})}$-Fourier-regular. 
Using Proposition \ref{prop:its-cool}, $f_{\rho'\eta}$
  is $2\tau'/(a\sqrt{\lambda})$-Fourier-regular with probability at least $\smash{1-a\sqrt{\lambda}/2}$.
It follows from Claim \ref{hard} and a~union~bound that, for an $\eta$ drawn from $\{-1,1\}^H$,
  $f_{\rho'\eta}$ is both 
  $(4\eps/3)$-balanced and $(2\tau^*/a)$-Fourier-regular 
  with probability at least $a\sqrt{\lambda}/2$ and thus,
  $\Pr[E_3\hspace{0.03cm}|\hspace{0.03cm}E_0\land E_1]\ge 1-1/n$, 
  by choosing~a large enough ${C_M}$.  

In the rest of the proof, we first show that all four events $E_0,E_1,E_2,E_3$ together are 
  sufficient to imply~that
  the subroutine returns a restriction $\eta$ such that 
  $f_{\rho' \eta}$ is $(\tau,\eps,\lambda')$-non-monotone,
  and then prove Claim \ref{hard}.
Combining everything together, we have that the subroutine returns either
  ``non-monotone'' or a restriction $\eta$ with the desired property with probability at least
$$
\Pr[E_0']\cdot (1-\delta)+\Pr[E_0]\cdot (1-O(\delta/\tau'^2))\cdot
  (1-O(\tau'^2/\sqrt{\lambda})-(1/n))>1-1/(100\log n),
$$
using $\Pr[E_0']+\Pr[E_0]\ge 1-\delta$ and by choosing a large enough constant $C_M$.

Now we prove the sufficiency of $E_0,E_1,E_2,E_3$.
Similar to the proof of Lemma \ref{lem:regularize-and-balance}, 
  we split the proof into two steps. We start with the following claim.

\begin{claim}\label{rrrr}
Suppose $f_{\rho' \eta}$ is both $(4\eps / 3)$-balanced and $(2\tau^*/a)$-Fourier-regular. Then assuming~$E_2$ and the choice of a large enough $C_M$, $\rho' \eta$ passes both tests in Step 4 and $f_{\rho' \eta}$ is $\tau$-weight-regular. 
\end{claim}

\begin{proof}
Since $f_{\rho' \eta}$ is $(2\tau^*/a)$-Fourier-regular, by choosing a large
  enough $C_M$, {\tt Check-Fourier-Regular} will output ``not regular'', given $E_2$. Since $f_{\rho' \eta}$ is $(4\eps / 3)$-balanced, {\tt Estimate-Mean} will return a value which is~at least $1 - 4\eps / 3 - \eps / 6 \geq 1 - 7\eps / 6$, given $E_2$. Therefore, $\rho'\eta$ passes both tests. On the other hand,~by Theorem~\ref{Dzindzalieta},  $f_{\rho'\eta}$ is $O(\tau^*/(a\eps))$-weight-regular, and thus $\tau$-weight-regular. 
\end{proof}

Given $E_3$, at least one of the $\eta$ sampled in Step 4 satisfies the 
 assumption of Claim \ref{rrrr}.
It~then follows from Claim \ref{rrrr} that $\rho'\eta$ of such an $\eta$ will pass both tests in Step 4.
Moreover, we~claim that $f_{\rho'\eta}$ is a $(\tau,\eps,\lambda')$-non-monotone LTF.
Here the only missing part is that $f_{\rho'\eta}$ has $\lambda'$-significant squared negative weights,
  which follows from  {(\ref{eq:1})}  and $w_i>0$ for all $i\in H$,~given $E_0$.   
(Indeed they imply that any restriction $\eta$ over $H$ has 
  this property, which we also use in the next claim.) 

We finish the proof of the sufficiency of $E_0,E_1,E_2,E_3$ with the following claim:

\begin{claim}
Assuming $E_2$ and the choice of a large enough $C_M$, any restriction $\eta$ on $H$ that passes both the {\tt Check-Fourier-Regular} test and the {\tt Estimate-Mean} test must be $(\tau, \eps, \lambda')$-non-monotone.
\end{claim}

\begin{proof}
Given $E_2$, we can conclude that a restriction $\rho'\eta$ that passes both tests satisfies that~$f_{\rho' \eta}$ is $\sqrt{C_M\tau^*}$-Fourier-regular  and $\eps$-balanced. 
As discussed earlier, $f_{\rho'\eta}$ always has $\lambda'$-significant squared negative weights.
By Theorem~\ref{Dzindzalieta} and the choice of a large enough $C_M$, $f_{\rho'\eta}$ is $\tau$-weight-regular. \end{proof}

Therefore we have shown that the four events $E_0,E_1,E_2,E_3$ together imply that the subroutine
  returns a restriction $\eta$ with $f_{\rho'\eta}$ being $(\tau,\eps,\lambda')$-non-monotone.
Finally we prove Claim \ref{hard}.
\begin{proof}[Proof of Claim~\ref{hard}]
Assume without loss of generality $\smash{\sum_{i\in B_t} w_i^2=1}$ and let $\smash{\theta'\ge 0}$ be 
  the threshold of~$f_{\rho'}$. Let
$$
\alpha=\sum_{i\in H} w_i^2\quad \text{and}\quad \beta=\sum_{i\in B_t\setminus H} w_i^2,
$$
with $\alpha+\beta=1$.
By (\ref{eq:1}) and ${E_0}$, we have $\beta\ge \lambda'=\lambda(1-(t+1)/(8\log n))=\Omega(\lambda)$ as
  $t\le 4\log n$.

For each $i\in B_t$, let $\bx_i$ denote an independent and unbiased $\{{-1},1\}$-variable 
  and let $$\bx=\sum_{i\in H} w_i\cdot \bx_i\quad \text{and}\quad \bx'=\sum_{i\in B_t\setminus H} w_i\cdot \bx_i.$$
 {Using the fact that $f_{\rho^{(t)}}$ was $\tau$-weight regular and (\ref{eq:1}), $f_{\rho'}$ is $O(\tau)$-weight regular}. Thus, by~Berry--Ess\'een, 
$\bx+\bx'$ is $\smash{O(\tau)}$-close to 
  an $\calN(0,1)$ Gaussian random variable, denoted by $\smash{\calG}$, $\smash{\bx}$ is $\smash{O(\tau/\sqrt{\alpha})}$-close to
  $\smash{\sqrt{\alpha}\hspace{0.03cm}\calG}$, and $\smash{\bx'}$ is $\smash{O(\tau/\sqrt{\beta})}$-close to $\smash{\sqrt{\beta}\hspace{0.03cm}\calG}$.
Since $f_{\rho'}$ is $0.96$-balanced, 
\begin{equation*}
\Pr\big[\calG\le \theta'\big]\le \Pr\big[\bx+\bx'\le \theta'\big]+O(\tau)\le 0.52+O(\tau)
\end{equation*}
and thus, $\theta'<0.06$.
Let $\theta^*>0$ be the threshold such that
$\Pr[\hspace{0.03cm}|\sqrt{\beta}\hspace{0.03cm}\calG|\le \theta^*
\hspace{0.03cm}]=1-3\eps/2$.

We will use the following inequality:
\begin{equation}\label{hahe}
\Pr\big[{f_{\rho'\eta}}\text{\ is $(4\eps/3)$-balanced}\big]\ge \Pr\big[\bx\in {[\theta' -\theta^*,\theta']}\big].
\end{equation}
The idea is that if $\sum_{i\in H}w_i\cdot x_i \in [\theta' - \theta^*, \theta']$, and we write 
\[ f_{\rho' \eta}(x') = \sign\left( \sum_{i \in B_{t}\setminus H} w_i\cdot x'_i - \left(-\sum_{i \in H} w_i \cdot x_i + \theta'\right)\right), \]
then the new threshold is nonnegative and at most $\theta^*$. Thus $|\E[f_{\rho' \eta}(\bx')]| = 1 - 2\Pr[f_{\rho' \eta}(\bx') = 1]$.
Using $\beta = \Omega(\lambda)$ and $\tau / \sqrt{\lambda} = o(\eps)$, we also have 
\[ \Pr[f_{\rho'\eta}(\bx') = 1] \geq \Pr[ \bx' \geq \theta^* ] \geq \Pr[\sqrt{\beta}\calG \geq \theta^*] - O(\tau/\sqrt{\beta}) = (3\eps / 4) - o(\eps).\] 
This implies that $f_{\rho' \eta}$ is $(3\eps / 2 - o(\eps))$-balanced
  and thus,  $(4\eps / 3)$-balanced. 

Finally we 
  bound the probability of $\bx\in [\theta'-\theta^*,\theta']$
  by considering the following two cases. 
\begin{enumerate}
\item[] \textbf{Case 1:} $\alpha\ge 0.02$. 
We have
$$
\hspace{-0.8cm}\Pr\big[\sqrt{\alpha}\hspace{0.03cm}\calG\in [\theta'-\theta^*,\theta']\big]
=\Pr\big[\calG\in [(\theta'-\theta^*)/\sqrt{\alpha},\theta'/\sqrt{\alpha}]\big] 
\ge \min(1/\sqrt{\alpha},\theta^*/\sqrt{\alpha})\cdot \Omega(1)=  {\Omega(\sqrt{\lambda})},
$$
as $\theta^*=\Omega(\sqrt{\lambda})$ (using our choice of $\theta^*$, $\eps<1/2$, and $\beta=\Omega(\lambda)$).
It follows that
$$
 \Pr\big[\bx\in {[\theta' -\theta^*,\theta']}\big]\ge 
\Pr\big[\sqrt{\alpha}\hspace{0.03cm}\calG\in [\theta'-\theta^*,\theta']\big]
  -O(\tau/\sqrt{\alpha})= {\Omega(\sqrt{\lambda})}.
$$

\item[] \textbf{Case 2:} $\alpha<0.02$ and thus, $\beta>0.98$.
Combining $\theta'<0.06$ 
  and $\theta^*> \sqrt{\beta}\cdot 0.31> 0.3$: 
\begin{align*}
\hspace{-0.7cm}\Pr\big[\bx\in [\theta'-\theta^*,\theta']\big]
\ge \Pr\big[\bx \in [0, \theta^* - \theta']\big] \ge \frac{1}{2} - \exp(-25(\theta^* - \theta')^2),
\end{align*}
by Hoeffding inequality.
Plugging in $\theta^* - \theta'> 0.24$, the probability above is $\Omega(1)$.
\end{enumerate}
Summarizing the two cases, $f_{\rho'\eta}$ is $(4\eps/3)$-balanced with
  probability at least ${\Omega(\sqrt{\lambda})}$. 
  \end{proof}
  
This finishes the proof of the lemma.
\end{proof}

Lemma \ref{key-lemma} follows directly from Lemmas~\ref{lem:a},~\ref{lem:FBR}, and \ref{lem:lastlem}.



\subsection{Final analysis of the algorithm}

\begin{theorem}
\label{thm:correct} 
The algorithm {\tt Mono-Test-LTF}$(f, \eps)$ correctly tests whether a given LTF is monotone or $\eps$-far from monotone.
\end{theorem}

\begin{proof}
The algorithm is one-sided because it outputs ``non-monotone'' only when an anti-monotone edge is found. The only interesting case is when the input LTF $f$ is $\eps$-far from monotone. Combining Lemmas~\ref{lem:regularize-and-balance}
and \ref{lem:main-procedure}, 
  the algorithm {\tt Mono-Test-LTF}$(f,\eps)$ outputs ``non-monotone'' with probability at least
  $(9/10) (81/100)>2/3$. This completes the proof.
\end{proof}


\begin{theorem} \label{thm:query-complexity}
The algorithm {\tt Mono-Test-LTF}$(f, \eps)$ makes $\tilde{O}\hspace{0.03cm}({\log^{ {42}}n}/{\eps^{90}})$ queries.
\end{theorem}

\begin{proof}
From Fact~\ref{query-comp-1}, the number of queries used by {\tt Regularize-and-Balance} is $\tilde{O}\hspace{0.03cm}({\log^{41} n}/{\eps^{90}})$, since the main bottle\-neck is the call to {\tt Find-Hi-Influence-Vars}. In {\tt Main-Procedure}, the bottleneck is the {$O(\log n)$} calls to {\tt Find-Hi-Influence-Vars} in {\tt Maintain-Regular-and-Balance}, each of query complexity $\tilde{O}\hspace{0.03cm}({\log^{41}n}/{\eps^{ {90}}})$, despite the slightly different parameters. {Note that we run the edge tester when there are fewer than ${1}/{\tau^2}$ many stars, so it makes $\tilde{O}\left( {\log^4 n}/{\eps^9}\right)$ many queries}.
\end{proof}

Theorem~\ref{thm:main} follows as an immediate consequence of Theorems~\ref{thm:correct} and~\ref{thm:query-complexity}.
-
\begin{flushleft}
\bibliographystyle{alpha}
\bibliography{allrefs,odonnell-bib}
\end{flushleft}
 
\end{document}